\documentclass[aps,prx,superscriptaddress,twocolumn,reprint,a4paper,noeprint,floatfix]{revtex4-2}
\usepackage[utf8]{inputenc}
\usepackage{amsmath}
\usepackage{amsthm}
\usepackage{amsfonts}
\usepackage{amssymb}
\usepackage[usenames,dvipsnames,svgnames,table]{xcolor}
\usepackage[USenglish]{babel}
\usepackage{hyperref}
\usepackage{graphicx}
\usepackage{extpfeil}

\allowdisplaybreaks
\DeclareMathOperator{\Tr}{Tr}
\let\Re\undefined
\DeclareMathOperator{\Re}{Re}
\let\Im\undefined
\DeclareMathOperator{\Im}{Im}

\DeclareMathOperator{\sgn}{sgn}
\DeclareMathOperator{\diag}{diag}

\newcommand{\Cl}{Cl}

\newtheorem{definition}{Definition}
\newtheorem{theorem}{Theorem}
\newtheorem{lemma}{Lemma}

\begin{document}

\title{Topological classification of non-Hermitian Hamiltonians with frequency dependence}

\author{Maximilian Kotz}
\email{maximilian.kotz@tu-dresden.de}
\affiliation{Institute of Theoretical Physics, Technische Universit\"at Dresden, 01062 Dresden, Germany}

\author{Carsten Timm}
\email{carsten.timm@tu-dresden.de}
\affiliation{Institute of Theoretical Physics, Technische Universit\"at Dresden, 01062 Dresden, Germany}
\affiliation{W\"urzburg-Dresden Cluster of Excellence ct.qmat, Technische Universit\"at Dresden, 01062 Dresden, Germany}

\date{April 13, 2023}

\begin{abstract}
We develop a topological classification of non-Hermitian effective Hamiltonians that depend on momentum and frequency. Such effective Hamiltonians are in one-to-one correspondence to single-particle Green's functions of systems that satisfy translational invariance in space and time but may be interacting or open. We employ \textit{K}-theory, which for the special case of noninteracting systems leads to the well-known tenfold-way topological classification of insulators and fully gapped superconductors. Relevant theorems for \textit{K}-groups are reformulated and proven in the more transparent language of Hamiltonians instead of vector bundles. We obtain 54 symmetry classes for frequency-dependent non-Hermitian Hamiltonians satisfying anti-unitary symmetries. Employing dimensional reduction, the group structure for all these classes is calculated. This classification leads to a group structure with one component from the momentum dependence, which corresponds to the non-Hermitian generalization of topological insulators and superconductors, and two additional parts resulting from the frequency dependence. These parts describe winding of the effective Hamiltonian in the frequency direction and in combined momentum-frequency space.
\end{abstract}

\maketitle

\section{Introduction}
\label{sc:introduction}

In recent years, topology has been established as an important organizing principle in condensed-matter physics, besides symmetry. A major breakthrough was the topological classification of insulators and fully gapped superconductors in ten symmetry classes (the ``tenfold way'') \cite{Zirnbauer, PhysRevB.55.1142, Schnyder_2008, Schnyder_2009, Kitaev_2009, Ryu_2010, Ludwig_2015}. This classification focused on effectively noninteracting Hamiltonians with translational invariance. For such models, one can construct the Bloch Hamiltonian $H(\mathbf{k})$, which is a map from the \textit{d}-dimensional unit cell in reciprocal space into the set of Hermitian matrices.
A vast body of literature on the topological classification of such Bloch Hamiltonians exists; for reviews see, for example, Refs.~\cite{Fruchart_2013, Ludwig_2015, RevModPhys.88.035005}.

The central question is whether two Bloch Hamiltonians with an energy gap, in $d$ spatial dimensions and satisfying certain symmetry constraints, can be continuously deformed into each other. The Hamiltonians for which this is possible form equivalence classes, which are, by construction, topological invariants. The equivalence classes form a  group. This group turns out to depend on the spatial dimension $d$ and on the symmetries. Unitary symmetries lead to a block-diagonal structure of the Hamiltonian. Each of the remaining blocks can have further anti-unitary symmetries, namely time-reversal symmetry or particle-hole symmetry. Taking into account whether these are absent or present, whether they square to the identity or minus the identity, and how they can be combined, one ends up with the aforementioned ten symmetry classes \cite{Zirnbauer, PhysRevB.55.1142}. For these classes, one can calculate the homotopy group depending on the spatial dimension $d$ \cite{Schnyder_2008, Schnyder_2009, Kitaev_2009, Ryu_2010, Ludwig_2015}. The resulting periodic table provides an organizing principle for gapped topological systems and naturally accommodates several known cases, for example the integer quantum Hall effect \cite{PhysRevLett.45.494, RevModPhys.58.519}, while dramatically broadening the scope.

Nontrivial values of topological invariants were found to be related to the existence of states localized at the $(d-1)$-dimensional surface of the system and to associated anomalous response to external fields \cite{PhysRevB.23.5632, PhysRevB.84.125132}. The quantized Hall conductivity is again a prime example.

The topological classification has been extended in various ways, e.g., to systems without an energy gap, where a classification of Fermi surfaces and nodal lines and points is constructed \cite{PhysRevLett.95.016405, PhysRevLett.110.240404, Matsuura_2013, PhysRevB.90.024516, RevModPhys.88.035005}. Another obvious question pertains to the extension to interacting systems. It turns out that the latter extension of the concept of topologically nontrivial systems is ambiguous, i.e., different notions of what makes a system nontrivial lead to different classifications. For example, one can \emph{define} topologically nontrivial states by their anomalous response to external fields \cite{PhysRevB.78.195424}. Other generalizations are based on surface states \cite{RevModPhys.88.035005, PhysRevB.83.075103}, entanglement \cite{1DEntange}, and the excitation gap in the bulk~\cite{PhysRevB.81.134509, PhysRevB.83.075103}.

One idea for the topological classification of interacting systems is to define an effective single-particle Hamiltonian $H_{\text{eff}}(\mathbf{k},\omega)$ that coincides with the Bloch Hamiltonian in the case of noninteracting systems. To that end, one can introduce the (retarded) single-particle Green's function~\cite{Mahan, altland_simons_2010}
\begin{align}
G^R_{nn'}(\mathbf{k},\omega) &= -i \int d^3r \int_0^\infty dt\,
  e^{-i\mathbf{k}\cdot\mathbf{r} + i\omega t} \nonumber \\
&\quad{}\times \big\langle \{ \Psi_n(\mathbf{r},t), \Psi^\dagger_{n'}(0,0) \} \big\rangle ,
\label{eq:DefGR}
\end{align}
where $\Psi_n(\mathbf{r},t)$ is a fermionic quantum-field operator at position $\mathbf{r}$ and time $t$ and the subscript $n$ contains all local degrees of freedom, such as spin, orbital, and basis site. $\{\bullet,\bullet\}$ is the anticommutator. The Green's function is a matrix on the space of the local degrees of freedom. This matrix can always be written in terms of an effective single-particle Hamiltonian as~\cite{PhysRevX.2.031008, BTS13, PhysRevB.98.035141, TRSPHS, PhysRevResearch.1.012003, Yoshida_2020, Lessnich_2021, RevModPhys.93.015005}
\begin{equation}
G^R(\mathbf{k},\omega) = \left[ \omega - H^R_\mathrm{eff}(\mathbf{k},\omega) \right]^{-1}
\label{eq:DefEffHam}
\end{equation}
(an identity matrix is suppressed). The advantage is that the Green's function remains well defined in the presence of interactions or driving, leading to an effective Hamiltonian $H^R_\mathrm{eff}$ that can be classified. In contrast to the Hamiltonian of a noninteracting system, this effective Hamiltonian is generally non-Hermitian. Additionally, it generally depends on frequency, not only on momentum. Previously, the frequency dependence has been ignored by considering $H^R_\mathrm{eff}$ at $\omega=0$ only. For closed interacting systems, all eigenvalues of the retarded effective Hamiltonian $H^R_\mathrm{eff}(\mathbf{k},\omega)$ have non-positive imaginary parts. This means that the only possible scenario is that of a real line gap (to be defined below).

Non-Hermitian Hamiltonians also appear in the context of open and driven quantum systems~\cite{Hatano_1996, Bender_2007, Rotter_2009, Longhi_2015, PhysRevX.8.031079, McDonald_2018, Ramos_2021, RevModPhys.93.015005, McDonald_2022, PhysRevB.106.L121102, roccati2023hermitian}. Driven systems permit non-Hermitian Hamiltonians with eigenvalues with positive or negative imaginary parts, even in the absence of explicit interactions. Positive (negative) imaginary parts correspond to amplification (damping) of eigenstates. A famous example is the Hatano--Nelson model \cite{Hatano_1996, RevModPhys.93.015005, Longhi_2015, PhysRevB.106.L121102}, which describes asymmetric hopping along a chain. This property of the spectrum also applies to the effective Hamiltonian constructed from the retarded Green's function \cite{McDonald_2018, McDonald_2022} and it persists in the presence of interactions \cite{PhysRevB.106.L121102}. To the best of our knowledge, the frequency dependence of the effective Hamiltonian has not been studied in this context.

In this paper, we study the consequences of the frequency dependence of $H^R_\mathrm{eff}(\mathbf{k},\omega)$. The frequency dependence does not simply increase the dimension from $d$ to $d+1$ since frequency and momentum behave differently under the relevant global symmetries. This has two consequences: First, there is a larger number of possible symmetry classes, and, second, nontrivial topology is possible due to the frequency dependence alone, the momentum dependence alone, and the combined dependence. This leads to a much richer topological classification.

The remainder of this paper is organized as follows. In Sec.\ \ref{sc:ClsymmClass}, we define and calculate the \textit{K}-groups of effective Hamiltonians. The calculation is based on existing ideas but uses the more transparent language of matrix-valued functions instead of vector bundles. We then consider the consequences of the frequency dependence for the symmetry classes, review the flattening procedure for non-Hermitian Hamiltonians as well as the dimensional reduction. We then present the resulting classification and discuss a few examples. For illustration, we analyze a toy model that exhibits nontrivial topology due to the frequency dependence in Sec.\ \ref{sc:example}. In Sec.\ \ref{sc:concl}, we summarize the results and draw some conclusions.

\section{Topological group structure}
\label{sc:ClsymmClass}

The goals of this section are to obtain the possible symmetry classes of frequency-dependent non-Hermitian Hamiltonians and to characterize the topological properties of such Hamiltonians belonging to these symmetry classes. The method involves several different mathematical concepts, which are discussed in the following.

\subsection{General properties of \textit{K}-groups}
\label{sc:Kgroup}

In this section, the \textit{K}-groups of Hermitian Hamiltonians are defined and some of their properties are presented. The general case of non-Hermitian Hamiltonians can be reduced to this case, as we will see below. The definition of the groups follows the definitions leading to the well-known topological insulators and superconductors. The term ``Hamiltonian'' is inspired by the free case but in this context can refer to any matrix-valued function, e.g., a correlation function, Green's function, or effective Hamiltonian as in Eq.~(\ref{eq:DefEffHam}). The mathematical part of this section is mainly inspired by Ref.\ \cite{Karoudi} but instead of dealing with \textit{K}-groups of vector bundles like in Kitaev's original approach \cite{Kitaev_2009}, we directly define the group for Hamiltonians to make the calculation more transparent.

\begin{definition}[Category $\mathcal{H}(X, \Omega)$]
$\mathcal{H}(X, \Omega)$ is a category. Its objects (called ``Hamiltonians'') are smooth maps from the manifold $X$ to Hermitian matrices with an energy gap (no eigenvalue equal to $0$ on $X$) and symmetries given by the set $\Omega$. Morphisms are smooth maps between the Hamiltonians.
\end{definition}

The category $\mathcal{H}(X, \Omega)$ is additive induced by the direct sum $\oplus$ of two Hamiltonians:
\begin{equation}
    H_1(\mathbf{k},\omega) \oplus  H_1(\mathbf{k},\omega) = \begin{pmatrix}
            H_1(\mathbf{k},\omega) & 0 \\
            0 &H_2(\mathbf{k},\omega) 
            \end{pmatrix} .
\end{equation}

\begin{definition}[Isomorphic Hamiltonians]
\label{definition:isomorphism}
Two objects of $\mathcal{H}(X, \Omega)$ are called isomorphic (denoted by ``$\,\approx$'') if they are homotopy equivalent and the corresponding homotopy function $H(\mathbf{k},\omega,t)$ for all $t$ has an energy gap and respects the symmetries in $\Omega$.
\end{definition}

Two Hamiltonians are generally not isomorphic if they are point-wise unitarily equivalent. 
In the following, just the symbol $\mathcal{H}$ is used, keeping in mind the dependence on $X$ and $\Omega$.

\begin{lemma}[Flattened Hamiltonian]
Every object of $\mathcal{H}$ is isomorphic to a Hamiltonian with eigenvalues $\pm 1$ \cite{Kitaev_2009}. 
\end{lemma}
    
\begin{proof}
Every Hamiltonian $H(\mathbf{k},\omega)$  can be represented as
\begin{equation}
H(\mathbf{k},\omega) = U(\mathbf{k},\omega) 
         \begin{pmatrix}
           E_1(\mathbf{k},\omega) & 0 & \dots\\
           0 & E_2(\mathbf{k},\omega) & \dots \\
           \vdots & \vdots & \ddots
         \end{pmatrix}
       U^\dagger(\mathbf{k},\omega) .
\end{equation}
Use as homotopy function a Hermitian matrix in the same representation but with eigenvalues  interpolating from their original values $E_i(\mathbf{k},\omega)$ to $\sgn E_i(\mathbf{k},\omega) = \pm 1$. It is straightforward to show that this homotopy function satisfies the conditions of Definition \ref{definition:isomorphism}~\cite{Kitaev_2009}.
\end{proof}

In this way, the dependence on $\mathbf{k}$ and $\omega$ is absorbed into the eigenvectors, while the eigenvalues have trivial behavior. The intuitive justification for this procedure is that the eigenvalues are real and thus cannot accumulate a nontrivial phase when $(\mathbf{k},\omega)$ is moved along a closed path. In the following, we deal with flattened Hamiltonians. These Hamiltonians square to the identity matrix.

\begin{definition}[Trivial elements]
The constant diagonal Hamiltonians with eigenvalues $\pm 1$ are called trivial and denoted by $0$.
\end{definition}

Note that the number of positive and negative eigenvalues is not of interest for this definition. However, some symmetry classes restrict the trivial elements to have an equal number of positive and negative eigenvalues.

\begin{lemma}
\label{lemma:consttriv}
Every constant Hamiltonian $H(\mathbf{k},\omega)=H$ over a connected space is isomorphic to a trivial element. 
\end{lemma}

\begin{proof} The flattened Hamiltonian is written in a diagonal form:
\begin{equation}
H = U     \begin{pmatrix}
            \pm 1 & 0     & \dots\\
            0     & \pm 1 & \dots \\
            \vdots & \vdots & \ddots
          \end{pmatrix}
        U^\dagger.
\end{equation}
We note that the diagonal matrix has to be everywhere the same because the space is connected.
There is a matrix $Y$ with the property $e^{iY} = U$. Now the  homotopy function can be set to
\begin{equation}
H(t) = U(t) 
          \begin{pmatrix}
            \pm 1 & 0     & \dots\\
            0     & \pm 1 & \dots \\
            \vdots & \vdots & \ddots
          \end{pmatrix}
        U(t)^\dagger,
\end{equation}
with $e^{itY} = U(t)$ and $t \in [0,1]$. The function has the same symmetries as $H$, is smooth, and is trivial for \mbox{$t=0$}.
\end{proof}

Note that the proof fails in the case of disconnected spaces because the number of positive and negative eigenvalues can differ in each connected component. Furthermore, it fails for $
(\mathbf{k},\omega)$-dependent Hamiltonians because there is no guarantee for the existence of a smooth function $Y(\mathbf{k},\omega)$. This can be seen by thinking about the logarithm of complex numbers.

\begin{lemma}
For every Hamiltonian $H(\mathbf{k},\omega) \in \mathcal{H}$: $H \oplus -H \approx 0$ \cite{Kitaev_2009}. 
\end{lemma}

\begin{proof}
We have
\begin{align}
H(\mathbf{k},\omega) \oplus -H(\mathbf{k},\omega)
     &= \begin{pmatrix}
          H(\mathbf{k},\omega) & 0 \\
          0   & -H(\mathbf{k},\omega) 
        \end{pmatrix} \nonumber \\
&\approx \begin{pmatrix}
          0 & i I \\
          -i I & 0 
        \end{pmatrix} ,
\end{align}
with identity matrices $I$, where the isomorphism holds because every interpolation between these matrices satisfies the symmetry condition. Furthermore, the energy gap does not close because both matrices square to the identity and anticommute. By Lemma \ref{lemma:consttriv} we then have
\begin{equation}
H(\mathbf{k},\omega) \oplus -H(\mathbf{k},\omega) \approx 0 .
\end{equation}
\end{proof}

\begin{definition}[The set $\Phi(\mathcal{H})$]
Let $H$ be an object of $\mathcal{H}$.
The isomorphism class of $H$ with respect to the isomorphism from Definition \ref{definition:isomorphism} is denoted by $\dot{H}$. 
The set of all such classes is called $\Phi(\mathcal{H})$.
\end{definition}

$\Phi(\mathcal{H})$ is an Abelian group induced by the additive structure of $\mathcal{H}$. The neutral element is given by the class containing the trivial Hamiltonian. The inverse of $\dot{H}$ is~$-\dot{H}$.

\begin{definition}[The set $K(\mathcal{H})$]
\label{definition:KH}
$K(\mathcal{H}) = \Phi(\mathcal{H})/{\sim}$ with the equivalence relation $\sim$ defined as follows: For two isomorphism classes $\dot{H}$ and $\dot{H'}$,
\begin{equation}
\dot{H} \sim \dot{H'} 
\end{equation}
iff there exist trivial Hamiltonians $T_1$, $T_2$ so that
\begin{equation}
H \oplus T_1 \approx H' \oplus T_2 ,
\end{equation}
where $H$ and $H'$ are arbitrary elements of $\dot{H}$ and $\dot{H'}$, respectively. Elements of $K(\mathcal{H})$ are denoted by $[H]$.
\end{definition}

This definition means that flat bands can be added to the Hamiltonian without changing the topological properties of interest. The definition also simplifies further derivations. $K(\mathcal{H})$ is an Abelian group since $\Phi(\mathcal{H})$ is. In the following, we denote the group by $K(X)$ to show the explicit dependence on the underlying space, which is the main focus of this section.

We aim to calculate $K(X)$ for the various symmetry classes and suitable spaces $X$. The space $X$ is the product of momentum space, which is the \textit{d}-dimensional torus $T^{d}$, and the frequency space, which is isomorphic to the real numbers $\mathbb{R}$. We replace the torus $T^d$ by the sphere $S^d$, thereby excluding the description of weak topological insulators and superconductors \cite{Fu_2007a, Kitaev_2009, Ryu_2010}. The spheres $S^d$ can also be motivated by starting from the momentum space $\mathbb{R}^d$ for continuum models and compactifying the infinite to a single point \cite{Ryu_2010, Budich_2013}. The frequency axis can also be compactified to a circle $S^1$ provided that for $\omega \to \pm \infty$ the effect of interactions becomes small so that the effective Hamiltonian approaches the free one \footnote{If this assumption is not satisfied one can disregard the frequency dependence, which leads to the theory in Ref.~\cite{TRSPHS}.}, perhaps except for a frequency-independent Hartree term. The frequency axis is naturally compactified for a time-periodic Hamiltonian within Floquet theory \cite{ASENS_1883_2_12__47_0, GRIFONI1998229, RevModPhys.89.011004}.
On the other hand, momentum and frequency are physically distinct quantities and there is no justification for compactifying the momentum-frequency space to $S^{d+1}$. The relevant spaces are thus $X = S^{d} \times S^{1}$. For quantum dots ($d=0$), the first set is not $S^0$ but contains only a single point and can be neglected. Note that quantum dots can still have a nontrivial frequency dependence. An example for this is shown in Sec.~\ref{sc:example}.

\begin{theorem}
\label{th:1}
 Let X be a compact space and A a compact connected retractable subspace of X. The following sequence is split exact \cite{Husemoller,Karoudi}:
\begin{equation}
0 \rightarrow K(X/A) \stackrel{\tilde{\beta}}{\longrightarrow} K(X) \stackrel{\tilde{\alpha}}{\longrightarrow} K(A) \rightarrow 0 ,
\end{equation}
where $X/A$ is the quotient space (one-point compactification) of $X$ and $A$ and the functors $\tilde{\alpha}$ and $\tilde{\beta}$ are induced by the natural inclusion $\alpha$ and restriction $\beta$,
\begin{equation}
A \stackrel{\alpha}{\longrightarrow} X \stackrel{\beta}{\longrightarrow} X/A .
\end{equation}
\end{theorem}

\begin{proof}
There are five statements to be proven:
\begin{itemize}
  \item[(i)] $\tilde{\beta}$ is injective,
  \item[(ii)] $\tilde{\alpha}$ is surjective,
  \item[(iii)] $\mathrm{Im}(\tilde{\beta})\subseteq \mathrm{Ker}(\tilde{\alpha})$,
  \item[(iv)] $\mathrm{Ker}(\tilde{\alpha})\subseteq \mathrm{Im}(\tilde{\beta})$,
  \item[(v)] splitting of the groups, i.e., $K(X) = K(X/A) \oplus K(A)$.
\end{itemize}

(i) Injectivity is equivalent to $\tilde{\beta}([H]) = 0 \Rightarrow [H] = 0$. Because of Definition \ref{definition:KH} there is in $K(X)$ a homotopy function between the Hamiltonian and the trivial element, perhaps after adding flat bands. One can choose a homotopy function that is constant in a neighborhood of $A$ because both $\tilde{\beta}([H])$ and the trivial element $0$ are constant on $A$. This implies that this homotopy function has a preimage in $K(X/A)$. This homotopy function interpolates in $K(X/A)$ between $H$ and the trivial element. 

(ii) Because $A$ is a retract of $X$ there is a  smooth function $\gamma$ from $X$ to $A$ that induces a map from $K(A)$ to $K(X)$. Furthermore, $\gamma \circ \alpha  = \mathrm{id}_{A}$ so that $\tilde{\gamma} \circ \tilde{\alpha} = \mathrm{id}_{K(A)}$. This implies the surjectivity of $\tilde{\alpha}$.

(iii) We consider the composition
\begin{equation}
\tilde{\alpha} \circ\tilde{\beta} ([H]) = [H\circ\beta\circ\alpha] .
\end{equation}
Since $\beta\circ\alpha$ is a constant function the map $H\circ\beta\circ\alpha$  is also constant. A constant Hamiltonian is isomorphic to a trivial Hamiltonian based on Lemma~\ref{lemma:consttriv}.

(iv) We choose a Hamiltonian $H' \in K(X)$ in the kernel of $ \tilde{\alpha}$:
\begin{equation}
\tilde{\alpha}([H']) = 0 .
\end{equation}  
One can use the preimage of the homotopy function in $K(A)$ to see that $H'|_{A}\approx 0$. Because of the smoothness of the Hamiltonians there is a neighborhood $V$ of $A$ on which $H'|_{V} \approx 0$. This implies that there is a homotopy function that interpolates between the Hamiltonian and the trivial element in this region. This homotopy-equivalent function $\tilde{H}$ is isomorphic to the original $H'$ but is constant on $V$. This new Hamiltonian has a preimage in $K(X/A)$.
 
(v) $\tilde{\gamma} \circ \tilde{\alpha} = \mathrm{id}_{K(A)}$ is also the condition for the splitting lemma \cite{Split}, which implies the splitting of the groups if the sequence is exact.
\end{proof}

We now choose $X=S^d \times S^1$ and $A=S^d \times {\omega_0} \cup {k_0} \times S^1$, where $k_0 \in S^d$ and $\omega_0 \in S^1$ are arbitrary but fixed elements. This implies $X/A=S^{d+1}$ \cite{Husemoller,Karoudi}. With Theorem \ref{th:1}, one obtains
\begin{equation}
K(S^d \times S^1) = K(S^{d+1}) \oplus K(S^d \times {\omega_0} \cup {k_0} \times S^1) ,
\label{KKK.5}
\end{equation}
where we can interpret $S^d$ as compactified momentum space and $S^1$ as compactified frequency space~\cite{DiffTopandQFT}.

\begin{theorem}
\label{th:2}
Let $X$ and $Y$ be compact connected spaces and $x_0 \in X$ and $y_0 \in Y$. Then the following relation holds \cite{DiffTopandQFT, Husemoller}: 
\begin{equation}
K(X\times {y_0} \cup {x_0} \times Y) = K(X) \oplus K(Y) .
\end{equation}
\end{theorem}

\begin{proof}
It is obvious that every element of $K(X\times {y_0} \cup {x_0} \times Y)$ induces an element in $K(X) \oplus K(Y)$.
For the opposite direction, the problem could be that a representative Hamiltonian of $K(X)$ at $x_0$ and a representative Hamiltonian of $K(Y)$ at $y_0$ are generally different. The Hamiltonians can be deformed to be constant in a neighborhood of $x_0$ and $y_0$. In this region, the exponential function can be used to interpolate between the original Hamiltonian and the trivial Hamiltonian at the point $x_0$ or $y_0$, see Lemma \ref{lemma:consttriv}.  
Now both Hamiltonians are equal, up to adding flat bands, and induce an element in $K(X\times {y_0} \cup {x_0} \times Y)$.
\end{proof}

Note that Theorem \ref{th:2} does not hold for $d=0$ since $S^0$ is not connected. From Eq.\ (\ref{KKK.5}) and Theorem \ref{th:2}, we obtain
\begin{equation}
K(S^d \times S^1) = K(S^d) \oplus K(S^1) \oplus K(S^{d+1})
\label{eq:splitK}
\end{equation}
for $d\ge 1$. Our classification leads to a term $K(S^d)$, which is well-known from the case without frequency dependence \cite{Kitaev_2009, Ryu_2010, Budich_2013}. The second term is new and results from the one-dimensional frequency space. The final term is also new and can lead to additional invariants from the $(d+1)$-dimensional space of momentum and frequency. These two terms result from interactions.

Note that most arguments use the fact that constant Hamiltonians are trivial, so locally every Hamiltonian looks trivial. There are two intuitive reasons why a Ha\-mil\-to\-ni\-an need not be trivial globally: (1) If $X$ is not connected, the different number of eigenvalues $\pm 1$ between the connected components is a topological number. (For Hamiltonians with particle-hole symmetry, there is always an equal number of positive and negative eigenvalues. However, in this case there is another characteristic number for unconnected spaces, see Ref.\ \cite{PhysRevB.90.165114}.) (2) In dimensions $d>0$, there is the problem that the inverse of the exponential function and thus $Y$ in the proof of Lemma \ref{lemma:consttriv} are not uniquely defined on the whole space of unitary matrices. Hence, although locally the Hamiltonian looks trivial, the whole map can be nontrivial. This is measured by winding and Chern numbers~\cite{PhysRevB.90.165114, Budich_2013}.

\subsection{Symmetry classes}
\label{sc:symClasses}

As noted in Sec.\ \ref{sc:introduction}, two global anti-unitary symmetries lead to the tenfold-way classification of frequency-independent Hermitian Hamiltonians \cite{Schnyder_2008, Schnyder_2009, Kitaev_2009, Ryu_2010, Ludwig_2015}. These symmetries and resulting classes  are now generalized to frequency-dependent and non-Hermitian Hamiltonians.

We follow the ideas of Ref.\ \cite{PhysRevX.9.041015}, where non-Hermitian Hamiltonians were treated but the frequency dependence was ignored. We also keep the notations for the symmetry classes used there. It is necessary to determine the possible sign change of the frequency argument $\omega$ of $H(\mathbf{k},\omega)$ under the symmetry transformation. Because of Eq.\ (\ref{eq:DefEffHam}) it is sufficient to analyze the frequency argument of the Green's function. For noninteracting fermions, one finds that $\omega$ changes sign under the particle-hole and chiral/sub-lattice transformations but not under time reversal \cite{Gurarie_2011}. The derivation is presented in some more detail in Appendix \ref{app:SymFreq}. The natural extension to interacting systems is to demand that the same relations persist. It is then easy to show using Eq.\ (\ref{eq:DefEffHam}) that the effective Hamiltonian satisfies the same symmetry relations as the Green's function.

The main difference between Hermitian and non-Hermitian Hamiltonians in the context of anti-unitary symmetries is that for non-Hermitian matrices $H$ we have $H \neq H^\dagger$ and $H^T \neq H^*$. This leads to a splitting of the symmetry conditions. Table \ref{tab:SymW} shows the possible global symmetries for frequency-dependent non-Hermitian Hamiltonians.

\begin{table*}[bt]
\caption{\label{tab:SymW}%
Symmetry conditions for non-Hermitian Hamiltonians with frequency dependence.
}
\begin{ruledtabular}
\begin{tabular}{lll}
pseudo-Hermiticity &    $\text{pH}$           & $U_{\text{pH}}\, H(\mathbf{k},\omega)^\dagger\, U_{\text{pH}}^\dagger = H(\mathbf{\mathbf{k}},\omega)$  \\
time-reversal symmetry &    $\text{TRS}$          & $U_{\text{TRS}}\, H(\mathbf{k},\omega)^*\, U_{\text{TRS}}^\dagger = H(-\mathbf{\mathbf{k}},\omega)$  \\
&    $\text{TRS}^\dagger$  & $U_{\text{TRS}^\dagger} H(\mathbf{k},\omega)^T\, U_{\text{TRS}^\dagger}^\dagger = H(-\mathbf{k},\omega)$  \\
particle-hole symmetry &    $\text{PHS}^\dagger$  & $U_{\text{PHS}^\dagger} H(\mathbf{k},\omega)^*\, U_{\text{PHS}^\dagger}^\dagger = -H(-\mathbf{k},-\omega)$  \\
&    $\text{PHS}$          & $U_{\text{PHS}}\, H(\mathbf{k},\omega)^T\, U_{\text{PHS}}^\dagger = -H(-\mathbf{k},-\omega)$  \\
chiral symmetry &     $\text{CS}$           & $U_{\text{CS}}\, H(\mathbf{k},\omega)^\dagger\, U_{\text{CS}}^\dagger = -H(\mathbf{k},-\omega)$  \\
sub-lattice symmetry &    $\text{SLS}$          & $U_{\text{SLS}}\, H(\mathbf{k},\omega)\, U_{\text{SLS}}^\dagger = -H(\mathbf{k},-\omega)$  
  \end{tabular}
\end{ruledtabular}
\end{table*}

In the absence of a frequency dependence, there is a unification of the $\text{TRS}$ and $\text{PHS}^\dagger$ symmetries \cite{PhysRevX.9.041015, TRSPHS}. The reason is that if $H(\mathbf{k})$ satisfies $\text{TRS}$, then $iH(\mathbf{k})$ satisfies $\text{PHS}^\dagger$ and vice versa. Hence, the classes containing $\text{TRS}$ and $\text{PHS}^\dagger$ are trivially related. In particular, the multiplication by $i$ just rotates the spectrum in the complex plane and leaves the eigenvectors unchanged. Analogously, $\text{pH}$ and $\text{CS}$ are unified. This leads to 38 fundamentally different symmetry classes \cite{PhysRevX.9.041015, TRSPHS, PhysRevB.99.235112}.
For frequency-dependent Hamiltonians, these unifications are not possible.

Symmetry classes can be distinguished by the number $N$ of anti-unitary generators taken from $\{\text{TRS}$, $\text{TRS}^\dagger$, $\text{PHS}^\dagger$, $\text{PHS}\}$. These four symmetries generalize the two anti-unitary symmetry of the Hermitian case and can each square to $\pm 1$ \cite{PhysRevX.9.041015}. The number $N$ can assume the following values:
\begin{itemize}
\item $N=0$: $6$ classes; no symmetry, one of $\text{pH}$, $\text{SLS}$, and $\text{CS}$, or all three. Two out of three is not possible because two of the symmetries generate the third. In the case of all three symmetries, $\text{SLS}$ and $\text{CS}$ may commute or anticommute, leading to a splitting into two classes.
\item $N=1$: $8$ classes; one of the four symmetries $\text{TRS}$, $\text{TRS}^\dagger$, $\text{PHS}^\dagger$, and $\text{PHS}$ is present and can square to $\pm 1$. The symmetries $\text{pH}$, $\text{SLS}$, and $\text{CS}$ cannot be present.
\item $N=2$: $6 \times 2^2 = 24$ classes; there are six ways to choose two of the four symmetries $\text{TRS}$, $\text{TRS}^\dagger$, $\text{PHS}^\dagger$, and $\text{PHS}$, and each one can square to $\pm 1$. The symmetries $\text{pH}$, $\text{SLS}$, and $\text{CS}$ are determined by the anti-unitary symmetries.
\item $N=3$: $2^4 = 16$ classes; if three of the symmetries $\text{TRS}$, $\text{TRS}^\dagger$, $\text{PHS}^\dagger$, and $\text{PHS}$ are present they generate the fourth one. Hence, the number of generators from this set is $N=3$ but all of them are present. The square of each of the four can be $\pm 1$ independently. The symmetries $\text{pH}$, $\text{SLS}$, and $\text{CS}$ are all present.
\item $N=4$: no additional classes since three of the anti-unitary symmetries $\text{TRS}$, $\text{TRS}^\dagger$, $\text{PHS}^\dagger$, and $\text{PHS}$ already generate all four.
\end{itemize}
In total, there is a $6+8+24+16 = 54$-fold classification. We denote the classes by their common symbols \cite{Schnyder_2008, Schnyder_2009, Kitaev_2009, Ryu_2010, Ludwig_2015}, adding a dagger ``$\dagger$'' if $\text{TRS}^\dagger$ or $\text{PHS}^\dagger$ are meant, and write the additional symmetry explicitly with the signs of the commutation relations in the subscript. If the class has two symmetries than the first sign in the subscript refers to $\text{TRS}^{(\dagger)}$ and the second one to $\text{PHS}^{(\dagger)}$. Table \ref{tab:symClasses} explicitly shows our conventions for the symmetry classes.

\begin{table*}
\caption{\label{tab:symClasses}Symmetry classification of frequency-dependent non-Hermitian Hamiltonians. Here, $0$ denotes that a symmetry is absent and $\pm 1$ means means that it is present and squares to plus or minus the identity. The classes $\text{AIII} + \text{SLS}_+$ and 
$\text{AIII} + \text{SLS}_-$ are distinguished by whether $\text{CS}$ and $\text{SLS}$ commute or anticommute, which is not reflected by the table entries.}
\begin{ruledtabular}
  \begin{tabular}{llllllll}
    Class    &  $\text{TRS}$    & $\text{TRS}^\dagger$ & $\text{PHS}$ & $\text{PHS}^\dagger$ & $\text{pH}$  & $\text{CS}$   & $\text{SLS}$ \\ \hline
    $\text{A}$     &  0        & 0             & 0     & 0             &  0    & 0      & 0     \\
    $\text{A} + \text{pH}$&  0        & 0             & 0     & 0             &  1    & 0      & 0     \\
    $\text{A} + \text{SLS}$&  0       & 0             & 0     & 0             &  0    & 0      & 1     \\
    $\text{AIII}$  &  0        & 0             & 0     & 0             &  0    & 1      & 0     \\ 
    $\text{AIII} + \text{SLS}_+$&  0  & 0             & 0     & 0             &  1    & 1      & 1     \\ 
    $\text{AIII} + \text{SLS}_-$&  0  & 0             & 0     & 0             &  1    & 1      & 1     \\ \hline
    $\text{AI}$     &  $+1$       & 0             & 0     & 0             & 0     & 0      & 0     \\
    $\text{D}$      &  0       & 0             & $+1$    & 0             & 0     & 0      & 0    \\
    $\text{AII}$    &  $-1$    & 0             & 0     & 0             & 0     & 0      & 0     \\
    $\text{C}$      &  0       & 0             & $-1$    & 0             & 0     & 0      & 0    \\
    $\text{AI}^\dagger$&  0    & $+1$           & 0     & 0             & 0     & 0      & 0     \\
    $\text{D}^\dagger$&  0     & 0             & 0     & $+1$            & 0     & 0      & 0    \\
    $\text{AII}^\dagger$& 0    & $-1$          & 0     & 0             & 0     & 0      & 0    \\ 
    $\text{C}^\dagger$&  0     & 0             & 0     & $-1$            & 0     & 0      & 0    \\ \hline
    $\text{BDI}$   &  $+1$      & 0             & $+1$    & 0             & 0     & 1      & 0     \\
    $\text{DIII}$  &  $-1$     & 0             & $+1$   & 0             & 0     & 1      & 0    \\
    $\text{CII}$   &  $-1$     & 0             & $-1$    & 0             & 0     & 1      & 0     \\
    $\text{CI}$    &  $+1$     & 0             & $-1$    & 0             & 0     & 1      & 0    \\
    $\text{BDI}^\dagger$&  0   & $+1$           & 0     & $+1$            & 0     & 1      & 0     \\
    $\text{DIII}^\dagger$& 0   & $-1$          & 0     & $+1$          & 0     & 1      & 0    \\
    $\text{CII}^\dagger$& 0    & $-1$          & 0     & $-1$            & 0     & 1      & 0     \\
    $\text{CI}^\dagger$&  0    & $+1$           & 0     & $-1$            & 0     & 1      & 0    \\
    $\text{AI}+\text{pH}_+$&  $+1$     & $+1$           & 0     & 0             & 1     & 0      & 0     \\
    $\text{AI}+\text{pH}_-$&  $+1$    & $-1$          & 0     & 0             & 1     & 0      & 0     \\
    $\text{AI}+\text{SLS}_+$&  $+1$    & 0             & 0     & $+1$            & 0     & 0      & 1     \\
    $\text{AI}+\text{SLS}_-$&  $+1$    & 0             & 0     & $-1$            & 0     & 0      & 1     \\
    $\text{D}+\text{pH}_+$&  0        & 0             & $+1$   & $+1$          & 1     & 0      & 0    \\
    $\text{D}+\text{pH}_-$&  0        & 0             & $+1$   & $-1$            & 1     & 0      & 0    \\
    $\text{D}+\text{SLS}_+$&  0       & $+1$           & $+1$    & 0             & 0     & 0      & 1    \\
    $\text{D}+\text{SLS}_-$&  0       & $-1$          & $+1$   & 0             & 0     & 0      & 1    \\
    $\text{AII}+\text{pH}_+$&  $-1$   & $-1$          & 0     & 0             & 1     & 0      & 0     \\
    $\text{AII}+\text{pH}_-$&  $-1$   & $+1$           & 0     & 0             & 1     & 0      & 0     \\
    $\text{AII}+\text{SLS}_+$&  $-1$  & 0             & 0     & $-1$            & 0     & 0      & 1     \\
    $\text{AII}+\text{SLS}_-$&  $-1$  & 0             & 0     & $+1$         & 0     & 0      & 1     \\
    $\text{C}+\text{pH}_+$& 0         & 0             & $-1$    & $-1$            & 1     & 0      & 0    \\
    $\text{C}+\text{pH}_-$& 0         & 0             & $-1$    & $+1$          & 1     & 0      & 0    \\
    $\text{C}+\text{SLS}_+$&  0       & $-1$            & $-1$    & 0             & 0     & 0      & 1    \\
    $\text{C}+\text{SLS}_-$&  0       & $+1$          & $-1$    & 0             & 0     & 0      & 1    \\ \hline
    $\text{BDI}+\text{pH}_{++}$& $+1$  & $+1$          & $+1$    & $+1$          & 1     & 1      & 1     \\
    $\text{BDI}+\text{pH}_{+-}$& $+1$  & $+1$           & $+1$   & $-1$            & 1     & 1      & 1     \\
    $\text{BDI}+\text{pH}_{-+}$& $+1$  & $-1$            & $+1$   & $+1$           & 1     & 1      & 1     \\
    $\text{BDI}+\text{pH}_{--}$& $+1$  & $-1$            & $+1$   & $-1$            & 1     & 1      & 1     \\
    $\text{DIII}+\text{pH}_{++}$&  $-1$ & $-1$            & $+1$   & $+1$          & 1     & 1      & 1    \\
    $\text{DIII}+\text{pH}_{+-}$&  $-1$ & $-1$            & $+1$   & $-1$            & 1     & 1      & 1    \\
    $\text{DIII}+\text{pH}_{-+}$&  $-1$ & $+1$          & $+1$  & $+1$         & 1     & 1      & 1    \\
    $\text{DIII}+\text{pH}_{--}$&  $-1$ & $+1$          &$+1$   & $-1$            & 1     & 1      & 1    \\
    $\text{CII}+\text{pH}_{++}$ &  $-1$ & $-1$            & $-1$    & $-1$            & 1     & 1      & 1     \\
    $\text{CII}+\text{pH}_{+-}$ &  $-1$ & $-1$            & $-1$    & $+1$         & 1     & 1      & 1     \\
    $\text{CII}+\text{pH}_{-+}$ &  $-1$ & $+1$         & $-1$    & $-1$            & 1     & 1      & 1     \\
    $\text{CII}+\text{pH}_{--}$ &  $-1$ & $+1$         & $-1$    & $+1$         & 1     & 1      & 1     \\
    $\text{CI}+\text{pH}_{++}$  & $+1$ & $+1$       & $-1$    & $-1$            & 1     & 1      & 1    \\
    $\text{CI}+\text{pH}_{+-}$  & $+1$ & $+1$         & $-1$    & $+1$          & 1     & 1      & 1    \\
    $\text{CI}+\text{pH}_{-+}$  & $+1$ & $-1$            & $-1$    & $-1$            & 1     & 1      & 1    \\
    $\text{CI}+\text{pH}_{--}$  & $+1$ & $-1$           & $-1$    & $+1$         & 1     & 1      & 1    
  \end{tabular}
\end{ruledtabular}
\end{table*}

\subsection{Flattening to Hermitian Hamiltonians}
\label{sc:flatt}

We use the standard definitions for different types of energy gaps \cite{PhysRevX.9.041015}:
For a point gap, the eigenvalues $E_i(\mathbf{k},\omega)$ of $H(\mathbf{k},\omega)$ satisfy $E_i(\mathbf{k},\omega) \neq 0$ for all $i$, $\mathbf{k}$, and $\omega$.
For a real line gap, $\Re E_i(\mathbf{k},\omega) \neq 0$ for all $i$, $\mathbf{k}$, and $\omega$.
For an imaginary line gap, $\Im E_i(k,\omega) \neq 0$ for all $i$, $\mathbf{k}$, and $\omega$.
Our strategy is to deform Hamiltonians with these different types of gaps into unitary, Hermitian, or anti-Hermitian ones, as shown in Fig.~\ref{pic:nonHermFlat}.
The cases with point and line gaps are discussed separately.

\begin{figure}
(a) Point gap\\
\includegraphics[height=4.5cm]{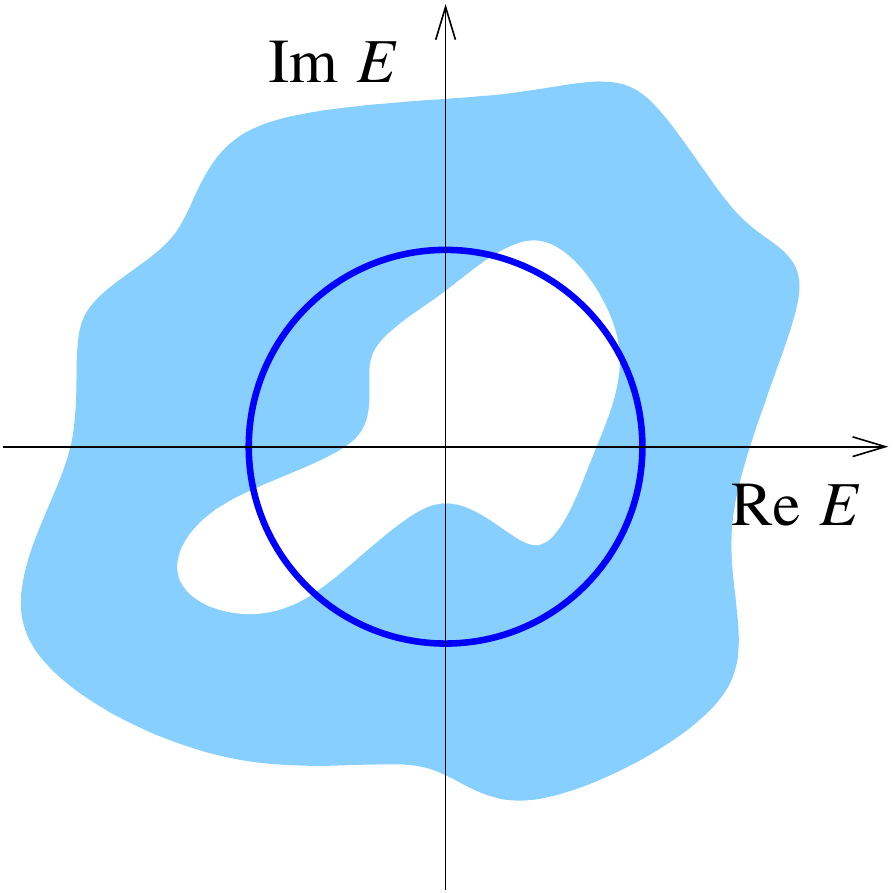}\\[1ex]
(b) Real line gap\\
\includegraphics[height=4.5cm]{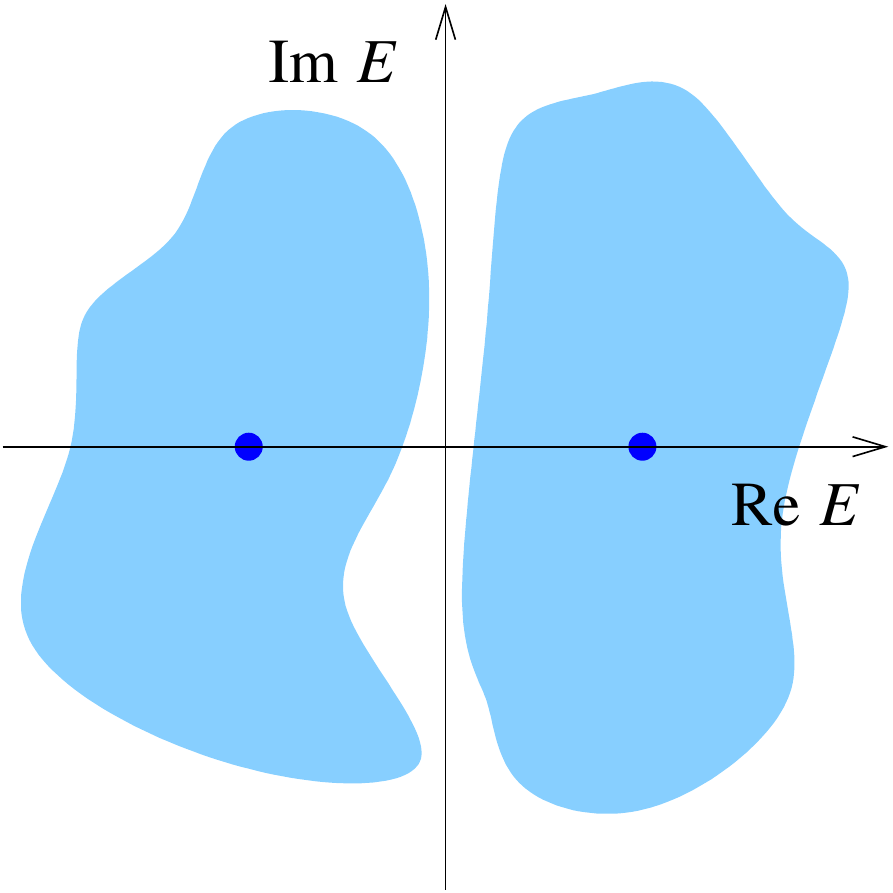}\\[1ex]
(c) Imaginary line gap\\
\includegraphics[height=4.5cm]{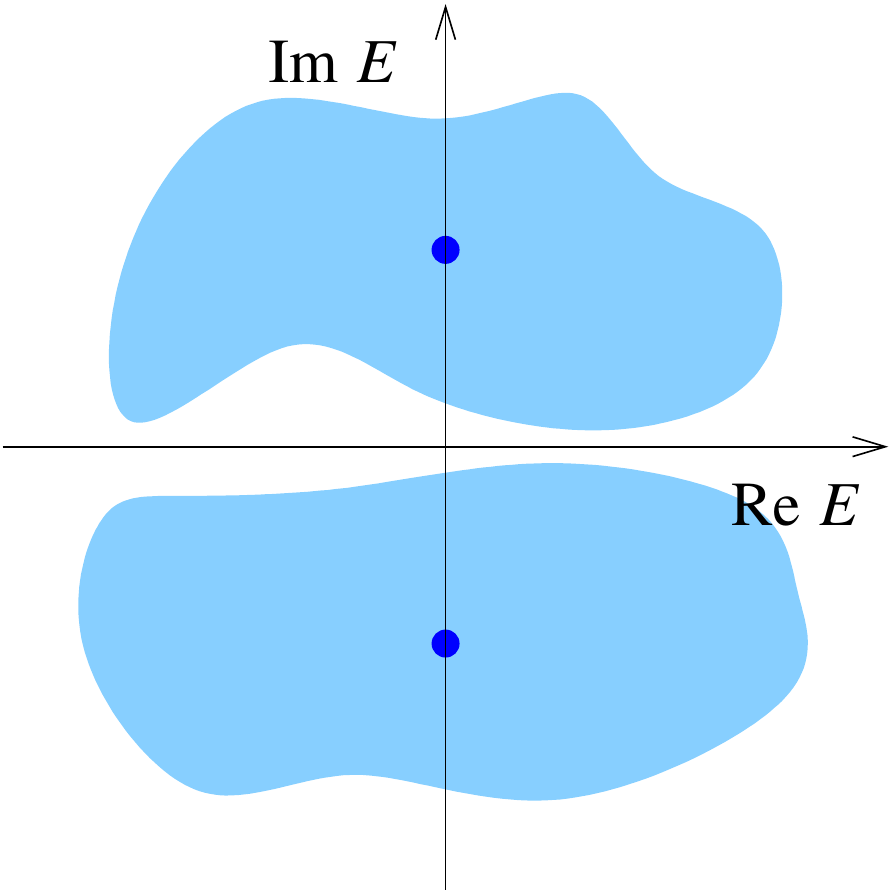}
\caption{\label{pic:nonHermFlat}Different types of gaps and flattening procedures for non-Hermitian Hamiltonians. (a) Flattening of a Hamiltonian with a point gap to a unitary one. (b) Flattening of a Hamiltonian with a real line gap to a Hermitian one with eigenvalues $\pm 1$. (c) Flattening of a Hamiltonian with a imaginary line gap to an anti-Hermitian one with eigenvalues~$\pm i$.}
\end{figure}


\begin{theorem}[Unitary flattening for point gaps]
Any Hamiltonian $H(\mathbf{k}, \omega)$ with a point gap is isomorphic to a unitary matrix $U(\mathbf{k}, \omega)$ with the same symmetries.
\label{theo:uniFlat}
\end{theorem}

The proof of this theorem is a straightforward generalization of the proof from Refs.\ \cite{PhysRevX.9.041015, PhysRevB.99.235112}. A different proof is presented in Ref.\ \cite{PhysRevX.8.031079}. The intuitive idea is that one can smoothly shift the eigenvalues in the complex plane so that they end up with unit modulus, see Fig.~\ref{pic:nonHermFlat}.

One can now define a \emph{Hermitian} Hamiltonian $\tilde{H}$ that can be classified~\cite{PhysRevX.9.041015, PhysRevB.99.235112},
\begin{equation}
    \tilde{H}(\mathbf{k},\omega) \equiv \left(\begin{array}{cc}
    0 & U(\mathbf{k},\omega) \\
    U^{\dagger}(\mathbf{k},\omega) & 0
    \end{array}\right).
    \label{eq:Huni}
\end{equation}
Clearly, $\tilde{H}^{2}(\mathbf{k},\omega)=1$.

If the original point-gapped Hamiltonian $H$ satisfies symmetries from Table \ref{tab:SymW}, the Hamiltonian $\tilde{H}$ satisfies the following symmetries~\cite{PhysRevX.9.041015}:
\begin{align}
\tilde{U}_{\text{pH}}\, \tilde{H}(\mathbf{k},\omega)\, \tilde{U}_{\text{pH}}^\dagger
  &= \tilde{H}(\mathbf{k},\omega),
\label{eq:sympoint.pH} \\
\mbox{with } \tilde{U}_{\text{pH}} &\equiv \left(\begin{array}{cc}
    0 & U_{\text{pH}} \\
    U_{\text{pH}} & 0
    \end{array}\right), \\
\tilde{U}_{\text{TRS}}\,\tilde{H}(\mathbf{k},\omega)^*\,\tilde{U}_{\text{TRS}}^\dagger
  &= \tilde{H}(-\mathbf{k},\omega) , \\
\mbox{with } \tilde{U}_{\text{TRS}} &\equiv \left(\begin{array}{cc}
    U_{\text{TRS}} & 0 \\
    0 & U_{\text{TRS}}
    \end{array}\right), \\
\tilde{U}_{\text{TRS}^\dagger}\,\tilde{H}(\mathbf{k},\omega)^*\,\tilde{U}_{\text{TRS}^\dagger}^\dagger
  &= \tilde{H}(-\mathbf{k},\omega) ,
\label{eq:sympoint.TRSd} \\
\mbox{with } \tilde{U}_{\text{TRS}^\dagger} &\equiv \left(\begin{array}{cc}
    0 & U_{\text{TRS}^\dagger} \\
    U_{\text{TRS}^\dagger} & 0
    \end{array}\right), \\
\tilde{U}_{\text{PHS}^\dagger}\,\tilde{H}(\mathbf{k},\omega)^*\,\tilde{U}_{\text{PHS}^\dagger}^\dagger
  &= -\tilde{H}(-\mathbf{k},-\omega) , \\
\mbox{with } \tilde{U}_{\text{PHS}^\dagger} &\equiv \left(\begin{array}{cc}
    U_{\text{PHS}^\dagger} & 0 \\
    0 & U_{\text{PHS}^\dagger}
    \end{array}\right), \\
\tilde{U}_{\text{PHS}}\,\tilde{H}(\mathbf{k},\omega)^*\,\tilde{U}_{\text{PHS}}^\dagger 
  &= -\tilde{H}(-\mathbf{k},-\omega) ,
\label{eq:sympoint:PHS} \\
\mbox{with } \tilde{U}_{\text{PHS}} &\equiv \left(\begin{array}{cc}
    0 & U_{\text{PHS}} \\
    U_{\text{PHS}} & 0
    \end{array}\right), \\
\tilde{U}_{\text{CS}}\,\tilde{H}(\mathbf{k},\omega)\,\tilde{U}_{\text{CS}}^\dagger 
  &= -\tilde{H}(\mathbf{k},-\omega) ,
\label{eq:sympoint:CS} \\
\mbox{with } \tilde{U}_{\text{CS}} &\equiv \left(\begin{array}{cc}
    0 & U_{\text{CS}} \\
    U_{\text{CS}} & 0
    \end{array}\right) , \\
\tilde{U}_{\text{SLS}}\,\tilde{H}(\mathbf{k},\omega)\,\tilde{U}_{\text{SLS}}^\dagger 
  &= -\tilde{H}(\mathbf{k},-\omega) , \\
\mbox{with } \tilde{U}_{\text{SLS}} &\equiv \left(\begin{array}{cc}
    U_{\text{SLS}} & 0 \\
    0 & U_{\text{SLS}}
    \end{array}\right),
\end{align}
and in all cases
\begin{align}
\tilde{U}_{\Sigma}\, \tilde{H}(\mathbf{k},\omega)\, \tilde{U}_{\Sigma}^\dagger 
  &= -\tilde{H}(\mathbf{k},\omega) ,
\label{eq:sigmaSym} \\
\mbox{with } \tilde{U}_{\Sigma} &\equiv \left(\begin{array}{cc}
    1 & 0 \\
    0 & -1
    \end{array}\right) .
\end{align}
In Eqs.\ (\ref{eq:sympoint.pH}), (\ref{eq:sympoint.TRSd}), (\ref{eq:sympoint:PHS}), and (\ref{eq:sympoint:CS}), we have used that $\tilde{H}(\mathbf{k},\omega)$ is Hermitian. Equation \eqref{eq:sigmaSym} guarantees that the mapping between $U(\mathbf{k},\omega)$ and $\tilde{H}(\mathbf{k},\omega)$ in Eq.~\eqref{eq:Huni} is bijective~\cite{PhysRevX.9.041015}.


\begin{theorem}[Unitary flattening for line gaps]
Any Hamiltonian $H(\mathbf{k},\omega)$ with a real line gap is isomorphic to a Hermitian Hamiltonian $\tilde{H}(\mathbf{k},\omega)$ with the same symmetries. Also, any Hamiltonian $H(\mathbf{k},\omega)$ with an imaginary line gap is isomorphic to an anti-Hermitian Hamiltonian $\check{H}(\mathbf{k},\omega)$ with the same symmetries.
\label{theo:lineFlat}
\end{theorem}

The proof of this theorem is again a straightforward generalization of the proof from Ref.\ \cite{PhysRevX.9.041015}. The first step is to diagonalize the Hamiltonian and flatten the eigenvalues to $\pm 1$ or $\pm i$, respectively, see Fig.\ \ref{pic:nonHermFlat}. Since the matrices needed for the diagonalization are generally not unitary it has to be shown that it is possible to deform the diagonalizing matrices into unitary ones. For this task, polar decomposition is used again. The nontrivial part of the proof is that this procedure is applicable only locally. However, it is shown in Ref.\ \cite{PhysRevX.9.041015} how to match different momentum regions to get a global deformation. The proof is not affected by the fact that the frequency $\omega$ transforms differently than the momentum $\mathbf{k}$. 

The flattened Hamiltonians have the same symmetries as the original ones. In the case of a real gap, one directly ends up with a Hermitian Hamiltonian $\tilde{H}(\mathbf{k},\omega)$ with the same symmetries as $H(\mathbf{k},\omega)$. In the case of an imaginary gap, one multiplies the anti-Hermitian Hamiltonian $\check{H}(\mathbf{k},\omega)$ by $i$ to obtain a Hermitian Hamiltonian $\tilde{H}(\mathbf{k},\omega) = i\check{H}(\mathbf{k},\omega)$ with symmetries that are trivially related to the symmetries of $\check{H}(\mathbf{k},\omega)$ and thus of $H(\mathbf{k},\omega)$.
We note in passing that one can analogously treat gaps shifted away from zero in the real (imaginary) direction if the symmetry class is compatible with adding a real (imaginary) constant times the identity operator to the effective Hamiltonian \cite{PhysRevX.8.031079}.

We see that the classification of non-Hermitian Hamiltonians can be reduced to a classification of flattened \emph{Hermitian} Hamiltonians for all three types of gaps. This is the reason why it was sufficient to consider Hamitian Hamiltonians in Sec.~\ref{sc:Kgroup}.

\subsection{Dimensional reduction}
\label{sc:Dred}

In this section, we calculate the group structure of Hamiltonians $H(\mathbf{k},\omega)$ that respect the global symmetries discussed in Sec.\ \ref{sc:flatt}. We formally collect these symmetries in the set $\Omega$.

As discussed in Sec.\ \ref{sc:Kgroup}, the base space of $\mathbf{k}$ is the \textit{d}-dimensional sphere $S^d$, where $d$ is the spatial dimension. The frequency space is compactified to $S^1$. Such Hamiltonians, or rather their equivalence classes, form the group $K(S^d\times S^1, \Omega)$. Using Theorems \ref{th:1} and \ref{th:2}, one obtains Eq.~(\ref{eq:splitK}). 

Our first goal is to relate the classes $K(S^D,\Omega)$ with $D = 1,d,d+1$ to the classes $K(S^0,\tilde{\Omega})$ describing zero-dimensional space and generally different symmetries $\tilde{\Omega}$. In this context, it is not important whether one dimension (i.e., frequency) behaves differently from the others under transformations. Hence, if we speak of momentum this includes the case that one of the components is actually a frequency.

The main idea is dimensional reduction \cite{PhysRevB.78.195424, Kitaev_2009, Ryu_2010, Stone_2010, PhysRevB.90.165114}. Dimensional reduction permits a physical interpretation on the basis of Dirac Hamiltonians. On the other hand, there is a more formal way of defining bijections between Hamiltonians in different dimensions. These approaches are strongly connected \cite{Stone_2010}. We here follow the more formal way because in the case with a frequency dependence the interpretation in terms of Dirac Hamiltonian could be confusing.

In order to determine the group $K(S^{D+1},\Omega)$, we parameterize the $(D+1)$-dimensional momentum $\mathbf{K} \in S^{D+1}$ by a polar angle $\theta \in \left[0,\pi\right]$ and a \textit{D}-dimensional vector $\mathbf{k}$, which contains all other components. If the inverse of the momentum is taken, $\mathbf{K} \to -\mathbf{K}$, there are two fixed points, $\mathbf{K}=0$ and $\mathbf{K}=\mathbf{K}_\infty$. Here, $\mathbf{K}_\infty$ is the compactified momentum onto which all boundary point of the unit cell $T^{D+1}$ of reciprocal space are mapped. The points $\theta=0,\pi$ correspond to the poles of the $(D+1)$-dimensional sphere $S^{D+1}$. Other than in Refs.\ \cite{Teo_2010, PhysRevB.90.165114}, we do not choose these poles to correspond to $\mathbf{K}=0$ and $\mathbf{K}=\mathbf{K}_\infty$ but rather choose them to lie in the equatorial hyperplane, i.e., they have $\theta = \pi/2$. Momentum inversion also inverts $\theta \to \pi-\theta$.

We define the action~\cite{Teo_2010, PhysRevB.90.165114}
\begin{equation}
S[H(\mathbf{k},\theta)] = \int d\theta\, d^{d}k\,
  \Tr (\partial_\theta H \partial_\theta H) .
\label{eq:varProb}
\end{equation}
The corresponding Euler--Lagrange equation with the constraint $H^2=1$ for the flattened Hamiltonian is
\begin{equation}
\partial_{\theta}^{2} H + m^2 H = 0 .
\label{eq:ELGvar}
\end{equation}
The action $S[H]$ is a positive semidefinite ``height'' functional of $H$. 
Going downhill from any flattened Hamiltonian results in a continuous trajectory connecting it to some flattened Hamiltonian $H$ that minimizes $S[H]$ and thus satisfies the Euler--Lagrange equation (\ref{eq:ELGvar}). The trajectory describes a continuous deformation so that the original and final Hamiltonians are isomorphic in the sense of Definition~\ref{definition:isomorphism}.

Equation (\ref{eq:ELGvar}) is solved by
\begin{equation}
H(\mathbf{k},\theta) = H_{1}(\mathbf{k}) \sin m \theta + H_{0} \cos m \theta ,
\end{equation}
where the constraint $H^2=1$ requires
\begin{align}
H_{0}^{2} = H_{1}(\mathbf{k})^{2} &= 1 ,
\label{eq:stdFormBed.a} \\
\left\{H_{0}, H_{1}(\mathbf{k})\right\} &= 0 .
\label{eq:stdFormBed.b}
\end{align}
The condition that $\theta = 0,\pi$ correspond to single $\mathbf{K}$ points implies that $H_0$ is a constant matrix and that $m$ is an integer. In the case of $m=0$, there is no $\mathbf{K}$ dependence so that the Hamiltonian is trivial. Trivial Hamiltonians are mapped to trivial ones during dimensional reduction. Solutions with $m<0$ can be mapped to solutions with $m>0$ by inverting the sign of $H_1(\mathbf{k})$. Values $m>1$ correspond to multiple windings in the $\theta$ direction. In these cases, one can add a flat band to the Hamiltonian and continuously deform the resulting Hamiltonian in such a way that the winding in the original part is reduced by unity, while the added part obtains this winding. The continuous deformation does not change the topological properties. By repeating this procedure, one can reduce the Hamiltonian to the case $m=1$.

Consequently, every Hamiltonian from $K(S^{D+1},\Omega)$ with nontrivial $\theta$ dependence can be written as
\begin{equation}
H(\mathbf{k},\theta) = H_{1}(\mathbf{k}) \sin \theta + H_{0} \cos \theta ,
    \label{eq:stdForm}
\end{equation}
up to isomorphism in the sense of Definition \ref{definition:isomorphism}. The Hamiltonian $H_1(\mathbf{k})$ is a representative of an element $[H_1(\mathbf{k})]$ of $K(S^D,\tilde{\Omega})$. The symmetries of $H_1(\mathbf{k})$ are collected in $\tilde{\Omega}$. They are generally different from the symmetries $\Omega$ of $H(\mathbf{k})$ for two reasons. First, Eqs.\ \eqref{eq:stdFormBed.a} and \eqref{eq:stdFormBed.b} impose additional constraints on $H_1(\mathbf{k})$. Second, the coordinate $\theta$, which parameterizes either a momentum component or frequency, may change sign under some of the symmetries satisfied by $H(\mathbf{k})$. This leads to $H_1(\mathbf{k})$ satisfying a different symmetry. Moreover, this new symmetry of $H_1(\mathbf{k})$ may be a conventional unitary one, which permits block diagonalization of $H_1(\mathbf{k})$. One then has to check which symmetries are satisfied for each block. We will give an example for this case below.

This construction defines a map from $K(S^{D+1},\Omega)$ to $K(S^D,\tilde{\Omega})$. Applying this map to the homotopy function connecting two elements from $K(S^{D+1},\Omega)$, we obtain a homotopy function connecting their images in $K(S^D,\tilde{\Omega})$. This implies that the  map respects the group structure.

It is straightforward to construct the inverse map
\begin{equation}
    K(S^D,\tilde{\Omega}) \xtofrom[B]{A} K(S^{D+1},\Omega) .
\end{equation}
This map induces an isomorphism \cite{Teo_2010}.
To prove this, note first that the map $A$ is well defined by Eq.\ (\ref{eq:stdForm}). Second, if $[H], [H'] \in K(S^D,\tilde{\Omega})$ are homotopic to each other, then the image of their homotopy function is a homotopy function in $K(S^{D+1},\Omega)$. Hence, $A$ respects the group structure. Third, one finds that $B \circ A = \mathrm{id}_{S^D}$ since the image of $A$ satisfies the Euler--Lagrange equation. Finally, $A \circ B = \mathrm{id}_{S^{D+1}}$ by construction.

To obtain $K(S^0,\tilde{\Omega})$, this procedure is repeated $D$ times. The benefit is that $K(S^0,\tilde{\Omega})$ can easily be related to known groups. The base space $S^0$ consists of two disconnected points, at each of which the Hamiltonian is just a constant matrix with symmetry constraints.

\subsection{Clifford algebras and classifying spaces}
 
After dimensional reduction, the Hamiltonian and the symmetry operations act locally at the two points of $S^0$. One can block diagonalize the Hamiltonian according to the conventional unitary symmetries. It is then sufficient to consider the Hamiltonian in a representative block. This Hamiltonian may have a number of symmetry constraints in the form of \emph{anticommuting} operators that also anticommute with the Hamiltonian. We choose a basis in such a way that the Hamiltonian at one of the two points of $S^0$ is the diagonal Hamiltonian $\diag(+1,\dots,-1,\dots)$. We are interested in all realizations of the Hamiltonian at the other point which are compatible with the symmetries. This is a problem of extending the Clifford algebra of the anticommuting symmetry operators by that Hamiltonian. All possible extensions of a Clifford algebra form its classifying space. Hence, the set of all possible Hamiltonians stands in one-to-one correspondence to the classifying space. Finally, we ask how many topologically distinct classes of such Hamiltonians exist, in the sense of there being no smooth deformation connecting them. This question is answered by the group $\pi_0$ of disconnected components of the classifying space. These groups are known \cite{PhysRevB.88.125129, PhysRevB.90.165114} and our procedure generalizes the one of Ref.\ \cite{PhysRevX.9.041015} to the 54 symmetry classes.

The complex Clifford algebra $\Cl_n$ [or $\Cl_n(\mathbb{C})$] has $n$ generators $e_i$ which satisfy
\begin{equation}
    \left\{e_{i}, e_{j}\right\}=2 \delta_{ij} .
\end{equation}
These generators constitute a basis $\{e_{1}^{n_{1}} \otimes e_{2}^{n_{2}} \otimes \ldots \otimes e_{n}^{n_{n}}\}_{n_{i}=0,1}$ of a vector space over $\mathbb{C}$. The maps
\begin{equation}
    \Cl_{n}=\left\{e_{1}, \ldots, e_{n}\right\} \rightarrow \Cl_{n+1}
    = \left\{e_0, e_{1}, \ldots, e_{n}\right\} ,
\end{equation}
where $e_0$ anticommutes with all $e_i$, define the so-called classifying spaces $\mathcal{C}_{n}$. Formally, the classifying space is the set of all possible representations of $e_0$ for given representations of $e_1$, \dots, $e_n$~\cite{PhysRevB.88.125129}.

The real Clifford algebra $\Cl_{p,q}$ [or $\Cl_{p,q}(\mathbb{R})$] has $p+q$ generators with the properties
\begin{align}
\left\{e_{i}, e_{j}\right\} = 0 \quad &\mbox{for } i\neq j, \\
e_{i}^{2} = -1 \quad &\mbox{for } i=1, \ldots, p , \\
e_{i}^{2} = +1 \quad &\mbox{for } i=p+1, \ldots, p+q .
\end{align}
These generators again constitute a basis of a vector space but in this case over $\mathbb{R}$.
The maps
\begin{align}
&\Cl_{0,q} = \left\{e_{1}, \ldots\right\} \rightarrow \Cl_{0,q+1}
  = \left\{e_0, e_{1}, \ldots \right\}, \\
&\mbox{with } e_0^{2} = +1
\end{align}
and $e_0$ anticommuting with all $e_i$, define the classifying spaces $\mathcal{R}_q$. One can then show \cite{Karoudi, Kitaev_2009} that the maps
\begin{align}
&\Cl_{p,q} = \left\{e_{1}, \ldots\right\} \rightarrow \Cl_{p+1,q}
  = \left\{e_0, e_{1}, \ldots \right\}, \\
&\mbox{with } e_0^{2} = -1
\end{align}
give $\mathcal{R}_{p+2-q}$ and the maps
\begin{align}
&\Cl_{p, q} = \left\{e_{1}, \ldots\right\} \rightarrow \Cl_{p, q+1}
  = \left\{e_0, e_{1}, \ldots\right\}, \\
&\mbox{with } e_0^{2} = +1
\end{align}
give $\mathcal{R}_{q-p}$ (in both cases with $e_0$ anticommuting with all $e_i$).
The classifying spaces have the Bott periodicity $\mathcal{C}_{p} \approx \mathcal{C}_{p+2}$ and $\mathcal{R}_{p} \approx \mathcal{R}_{p+8}$. The relevance of the classifying spaces is that their  $\pi_0$-groups are known \cite{Kitaev_2009, Ryu_2010, PhysRevB.88.125129, PhysRevB.90.165114}. 

In Appendix \ref{sc:zeroAZ}, we demonstrate that for frequency-independent Hermitian Hamiltonians the above procedure leads to the known results.

\subsection{\textit{K}-groups for frequency-dependent non-Hermitian Hamiltonians}

In this section, we show the results for the classification of frequency-dependent non-Hermitian Hamiltonians. These Hamiltonians belong to one of the 54 symmetry classes listed in Table \ref{tab:symClasses}. According to Eq.\ (\ref{eq:splitK}), the \textit{K}-groups for such Hamiltonians $H(\mathbf{k},\omega)$ in $d$ spatial or momentum dimensions consist of three parts,
\begin{equation}
K(S^d \times S^1) = K(S^d) \oplus K(S^1) \oplus K(S^{d+1}) ,
\label{eq:splitK.1}
\end{equation}
where $K(S^d)$ describes the topology in momentum space \cite{Kitaev_2009, Ryu_2010, Budich_2013, TRSPHS}, $K(S^1)$ refers to frequency space, and $K(S^{d+1})$ to combined momentum-frequency space. For the three cases of a point gap, a real line gap, and an imaginary line gap, the classification problem for $H(\mathbf{k},\omega)$ differently maps to the classification of Hermitian Hamiltonians with related symmetries, as described in Sec.\ \ref{sc:flatt}. Applying dimensional reduction, as described in Sec.\ \ref{sc:Dred}, the three \textit{K}-groups in Eq.\ (\ref{eq:splitK.1}) can be constructed.

The \textit{K}-groups for $d=1,2,3$ and all symmetry classes are shown in Table \ref{tab:rePoint} for a point gap, in Table \ref{tab:reLineRe} for a real line gap, and in Table \ref{tab:reLineIm} for an imaginary line gap. The columns for the frequency dependence are independent of the spatial symmetry $d$ but are repeated for clarity. (In the case of quantum dots, where real space and thus momentum space consists of a single point, this is the only contribution.) According to Eq.\ (\ref{eq:splitK.1}), the full \textit{K}-group is the direct sum of the three contribution from momentum, frequency, and their combination. Three examples that illustrate the procedure leading to Tables \ref{tab:rePoint}--\ref{tab:reLineIm} are presented below. Results for larger $d$ can be found in the same way.


The tables generalize the results in Ref.\ \cite{PhysRevX.9.041015}, where only the momentum dependence and the corresponding 38 symmetry classes were considered. The well-known results for the tenfold way \cite{Schnyder_2008, Schnyder_2009, Kitaev_2009, Ryu_2010, Ludwig_2015} are contained in Table \ref{tab:reLineRe} for a real line gap \footnote{We do not employ the notation ``$2\mathbb{Z}$'' which is sometimes used to indicate that a naturally defined topological invariant is even. $2\mathbb{Z}$ is isomorphic to $\mathbb{Z}$.}. Note that the entries in Table \ref{tab:reLineRe} for the momentum dependence are the same for symmetry classes that only differ by a dagger.
As noted in Sec.\ \ref{sc:introduction}, the case of a real line gap applies to the effective Hamiltonian extracted from the retarded Green's function of an interacting closed systems. For open systems, all three cases are possible.

Already for the momentum dependence \cite{PhysRevX.9.041015}, direct sums $\mathbb{Z}_2 \oplus \mathbb{Z}_2$ and $\mathbb{Z} \oplus \mathbb{Z}$ appear. Now, more complicated groups occur due to the combination of three contributions, i.e., because of including the frequency dependence. However, all groups are isomorphic to direct sums of potentially multiple copies of the groups $\mathbb{Z}$ and $\mathbb{Z}_2$. For example, for the class $\mathrm{AI}$ with $d=1$ and a point gap we find $\mathbb{Z} \oplus \mathbb{Z}_2 \oplus \mathbb{Z}_2$. Moreover, a novel feature appears for the class $\mathrm{AIII} + \mathrm{SLS}_+$ in odd dimensions $d$ and a line gap of either type: We find the \textit{K}-group $\mathbb{Z} \oplus \mathbb{Z} \oplus \mathbb{Z} \oplus \mathbb{Z}$ for the combined contribution alone. The derivation is sketched in Sec.\ \ref{sub.AIIISLSp}. The origin of the quadruple $\mathbb{Z}$ can be traced back to the frequency dependence of $\text{CS}$ and $\text{SLS}$, because of which dimensional reduction leads to a zero-dimensional Hamiltonian with two independent conventional unitary symmetries.


\subsubsection{Class $\text{A}$ in three dimensions with a point gap}
\label{sub:Ad3}

As the first example, we consider the case of class $\mathrm{A}$ for $d=3$ spatial dimensions, assuming a point gap.
Class $\mathrm{A}$ has no symmetries. Flattening the Hamiltonian $H(\mathbf{k},\omega)$ following Sec.\ \ref{sc:flatt}, we obtain a Hermitian Hamiltonian
\begin{equation}
\tilde{H}(\mathbf{k},\omega) \equiv \left(\begin{array}{cc}
    0 & U(\mathbf{k},\omega) \\
    U^{\dagger}(\mathbf{k},\omega) & 0
  \end{array}\right)
\end{equation}
with the additional symmetry $\Sigma$ described by
\begin{equation}
\tilde{U}_{\Sigma}\, \tilde{H}(\mathbf{k},\omega)\, \tilde{U}_{\Sigma}^\dagger
  = -\tilde{H}(\mathbf{k},\omega) ,
\end{equation}
where $\Sigma^2=1$.

We start with the contribution $K(S^1)$ from the frequency dependence. We write the Hamiltonian in the form of Eq.\ (\ref{eq:stdForm}), ignoring the momentum dependence,
\begin{equation}
H(\theta) = H_1 \sin\theta + H_0 \cos\theta ,
\label{eq:Ad3.Hreduction}
\end{equation}
where $\theta$ of course parameterizes the frequency. The two matrices $H_1$ and $H_0$ anticommute and square to the identity by construction. Furthermore, $H_1$ and $H_0$ anticommute with $\Sigma$.

Hence, the dimensionally-reduced Hamiltonian $H_1$ is subject to two constraints, namely it has to anticommute with $\Sigma$ and with $H_0$. In analogy to the standard Altland--Zirnbauer classes, see Appendix \ref{sc:zeroAZ}, we choose the anticommuting generators
\begin{align}
e_1 &= \Sigma , \\
e_2 &= H_0 ,
\end{align}
forming the Clifford algebra $\Cl_2$. The generator
\begin{equation}
e_0 = H_1
\end{equation}
extends the algebra,
\begin{equation}
\Cl_{2}=\left\{e_{1}, e_2\right\}
   \rightarrow \Cl_{2+1}=\left\{e_0, e_{1}, e_2 \right\} ,
\end{equation}
this defines the classifying space $\mathcal{C}_{2}\approx\mathcal{C}_{0}$. The resulting \textit{K}-group is $\mathbb{Z}$.

The second contribution $K(S^3)$ stems from the momentum dependence. We have to reduce the dimension from $d=3$ to $d=0$ in three steps. The result is an algebra generated by the symmetry operator $\Sigma$, the zero-dimensional Hamiltonian $H_1$, and three generators $H_0$, $H_0'$, and $H_0''$, which appeared due to the dimensional reduction. All of these generators anticommute and square to the identity. We choose $e_1 = \Sigma$, $e_2 = H_0$, $e_3 = H_0'$, and $e_4 = H_0''$ and add the generator $e_0 = H_1$, which leads to the extension $\Cl_{4} \rightarrow \Cl_{5}$. The classifying space is $\mathcal{C}_4 \approx \mathcal{C}_0$, implying that this part also contributes a $\mathbb{Z}$ group.

The last contribution $K(S^4)$ stems from the combined momentum and frequency dependence. Upon dimensional reduction, there arise three generators $H_0$, $H_0'$, and $H_0''$ from the momentum directions and one generator $H_0'''$ from the frequency direction. The corresponding extension is $\Cl_5 \rightarrow \Cl_6$. The classifying space is $\mathcal{C}_5 \approx \mathcal{C}_{1}$, leading to a trivial \textit{K}-group.

In conclusion, the entries for class $\text{A}$ and dimension $d=3$ in Table \ref{tab:rePoint} are $\mathbb{Z}$ for the momentum part, $\mathbb{Z}$ for the frequency part, and $0$ for the combined part. Using Eq.\ (\ref{eq:splitK}), we obtain the full group $\mathbb{Z} \oplus \mathbb{Z}$.

\subsubsection{Class $\text{AIII}$ in one dimension with  a point gap}

As a second example, we consider class $\mathrm{AIII}$ for $d=1$, assuming a point gap. We will see that an additional doubling of the group occurs. Class $\mathrm{AIII}$ has chiral symmetry $\text{CS}$. Flattening $H(k,\omega)$, we obtain a Hermitian Hamiltonian $\tilde{H}(k,\omega)$ with chiral symmetry,
\begin{equation}
\tilde{U}_\mathrm{CS}\, \tilde{H}(k,\omega)\, \tilde{U}_\mathrm{CS}^\dagger
  = - \tilde{H}(k,-\omega) ,
\end{equation}
and the additional symmetry
\begin{equation}
\tilde{U}_{\Sigma}\, \tilde{H}(k,\omega)\, \tilde{U}_{\Sigma}^\dagger
  = - \tilde{H}(k,\omega) .
\end{equation}
$\mathrm{CS}$ and $\Sigma$ square to the identity and anticommute.

The contribution from the frequency dependence requires one reduction according to Eq.\ (\ref{eq:Ad3.Hreduction}). The new aspect is that under $\text{CS}$, $\omega$ changes sign so that $\theta$ maps to $\pi-\theta$ and $\cos\theta$ changes sign. Hence, $H_0$ does not change sign under $\text{CS}$, i.e., $H_0$ and $\tilde{U}_\mathrm{CS}$ commute. $H_1$ retains $\mathrm{CS}$ and $\Sigma$ and thus anticommutes with $\tilde{U}_\mathrm{CS}$ and $\tilde{U}_\Sigma$, as well as with $H_0$. Since $H_0$ and $\tilde{U}_\mathrm{CS}$ commute they cannot both be generators of a Clifford algebra. However, $H_1$ commutes with their product $H_0 \tilde{U}_\mathrm{CS}$, which describes a conventional unitary symmetry. The unitary transformation that diagonalizes $H_0 \tilde{U}_\mathrm{CS}$ block diagonalizes $H_1$, resulting in two sectors corresponding to the two eigenvalues $\pm 1$ of $H_0 \tilde{U}_\mathrm{CS}$. $\Sigma$ and $H_0$ also commute with $H_0 \tilde{U}_\mathrm{CS}$ and are simultaneously block diagonalized. Therefore, the two blocks of $H_1$ are zero-dimensional Hamiltonians that anticommute with corresponding blocks of $\Sigma$ and $H_0$. We choose $e_1 = \Sigma$ and $e_2 = H_0$ and add $e_0 = H_1$ (or more correctly their diagonal blocks). The corresponding extension is $\Cl_2 \rightarrow \Cl_3$ with the classifying space $\mathcal{C}_2 \approx \mathcal{C}_0$ and group $\mathbb{Z}$. Since we have two sectors, the \textit{K}-group for the frequency dependence is doubled to $\mathbb{Z} \oplus \mathbb{Z}$.

The contribution from the momentum dependence also requires one reduction, which produces matrices $H_0$ and $H_1$ like in Sec.\ \ref{sub:Ad3}. But now $\theta$ is a momentum component and thus does not change under $\text{CS}$. Hence, $H_1$ inherits the symmetries $\Sigma$ and $\text{CS}$. We take $e_1 = \Sigma$, $e_2 = \mathrm{CS}$, and $e_3 = H_0$ and add $e_0 = H_1$. The corresponding extension is $\Cl_3 \rightarrow \Cl_4$ with the classifying space $\mathcal{C}_3 \approx \mathcal{C}_1$ and trivial \textit{K}-group.

The contribution from the combined momentum and frequency dependence requires two reductions according to Eq.\ (\ref{eq:Ad3.Hreduction}). Let us first take $\theta$ to represent the frequency. By the same argument as above, $H_1$ splits into two sectors and anticommutes with two matrices, $\Sigma$ and $H_0$. For each sector, reduction in the momentum direction generates a new matrix $H_0'$. The symmetry of $H_1$ is not changed. We choose $e_1 = \Sigma$, $e_2 = H_0$, and $e_3 = H_0'$ and add $e_0 = H_1$. The corresponding extension is $\Cl_3 \rightarrow \Cl_4$ with the classifying space $\mathcal{C}_3 \approx \mathcal{C}_1$ and trivial \textit{K}-group. The doubling of a trivial group because of the two sectors does not change it.

\subsubsection{Class $\text{AIII} + \text{SLS}_+$ in one dimension with a real line gap}
\label{sub.AIIISLSp}

The third example, the class $\text{AIII} + \text{SLS}_+$ for $d=1$ with a real line gap, involves a class that is only relevant for non-Hermitian Hamiltonians. Since we are considering a real line gap, the flattened Hamiltonian $\tilde{H}(k,\omega)$ is Hermitian. The Hamiltonian satisfies $\mathrm{CS}$ and $\mathrm{SLS}$, i.e., $\tilde{H}(k,\omega)$ anticommutes with two unitary matrices $U_\mathrm{CS}$ and $U_\mathrm{SLS}$, see Table \ref{tab:SymW}, which themselves commute.

The frequency contribution requires one reduction according to Eq.\ (\ref{eq:Ad3.Hreduction}). Under $\mathrm{CS}$ and $\mathrm{SLS}$, $\omega$ and $\cos\theta$ change sign. Thus $H_0$ does not change sign under either transformation and $H_0$ commutes with $U_\mathrm{CS}$ and $U_\mathrm{SLS}$, which themselves also commute. $H_1$ anticommutes with all three and commutes with the product of any two. Two products, say $H_0 U_\mathrm{CS}$ and $H_0 U_\mathrm{SLS}$, are independent. We can diagonalize them simultaneously and the corresponding transformation block diagonalizes $H_1$, leading to four blocks according to the four possible combinations of eigenvalues $\pm 1$. $H_0$ commutes with $H_0 U_\mathrm{CS}$ and $H_0 U_\mathrm{SLS}$ and is simultaneously block diagonalized. We choose $e_1 = H_0$ and add $e_0 = H_1$ (or more correctly their diagonal blocks). The extension is $\Cl_1 \rightarrow \Cl_2$ with classifying space $\mathcal{C}_1$ and trivial \textit{K}-group.

The momentum contribution also requires one reduction. $\theta$ does not change under $\mathrm{CS}$ or $\mathrm{SLS}$. $H_1$ thus inherits the anticommutation with $U_\mathrm{CS}$ and $U_\mathrm{SLS}$. $\mathrm{CS}$ and $\mathrm{SLS}$ commute and thus cannot both be generators of a Clifford algebra. Their product leads to a conventional unitary symmetry $\mathrm{CS}\times\mathrm{SLS}$, which can be used to block diagonalize $H_1$ and $H_0$ into two blocks. We choose (the diagonal blocks of) $e_1 = H_0$ and $e_0 = H_1$. The extension is $\Cl_1 \rightarrow \Cl_2$ with classifying space $\mathcal{C}_1$ and trivial \textit{K}-group.

The combined contribution can be obtained by first reducing the frequency dependence. As noted above, this leads to four sectors with Hamiltonians $H_1$ that anticommute with $H_0$. For each sector, reduction in the momentum direction generates a new $H_0'$ without changing the symmetry of $H_1$. We choose (the diagonal blocks of) $e_1 = H_0$, $e_2 = H_0'$, and $e_0 = H_1$. The extension is $\Cl_2 \rightarrow \Cl_3$ with classifying space $\mathcal{C}_2 \approx \mathcal{C}_0$ and group $\mathbb{Z}$. It is quadrupled because of the four sectors, leading to $\mathbb{Z} \oplus \mathbb{Z} \oplus \mathbb{Z} \oplus \mathbb{Z}$.

\begin{table*}
\caption{\label{tab:rePoint}
\textit{K}-groups for frequency-dependent non-Hermitian Hamiltonians in $d=1,2,3$ spatial dimensions for the case of a \emph{point gap}. For each value of $d$, the three columns show the \textit{K}-groups describing the momentum dependence (green, left column), the frequency dependence (blue, middle column), and the combined momentum and frequency dependence (red, right column).}
\begin{ruledtabular}
\begin{tabular}{l@{\qquad}lll@{\qquad}lll@{\qquad}lll}
    & & $d=1$ & & & $d=2$ & & & $d=3$ & \\
    Class & $K(S^d)$ & $K(S^1)$ & $K(S^{d+1})$ & $K(S^d)$ & $K(S^1)$ & $K(S^{d+1})$ & $K(S^d)$ & $K(S^1)$ & $K(S^{d+1})$ \\ \hline
    $\text{A}$             & $\color{DarkGreen}\mathbb{Z}$ & $ \color{blue}\mathbb{Z}$ & $ \color{red}0$             
                    & $\color{DarkGreen}0$ & $ \color{blue}\mathbb{Z}$ & $ \color{red}\mathbb{Z}$                                                                       
                    & $\color{DarkGreen}\mathbb{Z}$ & $ \color{blue}\mathbb{Z}$ & $ \color{red}0$                \\
    $\text{A} + \text{pH}$        & $\color{DarkGreen}0$ & $ \color{blue}0$ & $ \color{red}\mathbb{Z}$             
                    & $\color{DarkGreen}\mathbb{Z}$ & $ \color{blue}0$ & $ \color{red}0$                                                                                
                    & $\color{DarkGreen}0$ & $ \color{blue}0$ & $ \color{red}\mathbb{Z}$                  \\
    $\text{A} + \text{SLS}$       & $\color{DarkGreen}0$ & $ \color{blue}0$ & $ \color{red}\mathbb{Z}$             
                    & $\color{DarkGreen}\mathbb{Z}\oplus\mathbb{Z}$ & $ \color{blue}0$ & $ \color{red}0$                                                                                
                    & $\color{DarkGreen}0$ & $ \color{blue}0$ & $ \color{red}\mathbb{Z}$                  \\
    $\text{AIII}$          & $\color{DarkGreen}0$ & $ \color{blue}\mathbb{Z}\oplus\mathbb{Z}$ & $\color{red}0$
                    & $\color{DarkGreen}\mathbb{Z}$ & $ \color{blue}\mathbb{Z}\oplus\mathbb{Z}$ & $ \color{red}\mathbb{Z}\oplus\mathbb{Z}$
                    & $\color{DarkGreen}0$ & $ \color{blue}\mathbb{Z}\oplus\mathbb{Z}$ & $ \color{red}0$                \\
\end{tabular}
\end{ruledtabular}
\end{table*}

\begin{table*}
\caption{\label{tab:reLineRe}
\textit{K}-groups for frequency-dependent non-Hermitian Hamiltonians in $d=1,2,3$ spatial dimensions for the case of a \emph{real line gap}. For each value of $d$, the three columns show the \textit{K}-groups describing the momentum dependence (green, left column), the frequency dependence (blue, middle column), and the combined momentum and frequency dependence (red, right column).}
\begin{ruledtabular}
\begin{tabular}{l@{\qquad}lll@{\qquad}lll@{\qquad}lll}
    & & $d=1$ & & & $d=2$ & & & $d=3$ & \\
    Class & $K(S^d)$ & $K(S^1)$ & $K(S^{d+1})$ & $K(S^d)$ & $K(S^1)$ & $K(S^{d+1})$ & $K(S^d)$ & $K(S^1)$ & $K(S^{d+1})$ \\ \hline
    $\text{A}$             & $\color{DarkGreen}0 $ & $ \color{blue}0 $ & $ \color{red}\mathbb{Z}$                  
                    & $\color{DarkGreen}\mathbb{Z} $ & $ \color{blue}0 $ & $ \color{red}0$                             
                    & $\color{DarkGreen}0 $ & $ \color{blue}0 $ & $ \color{red}\mathbb{Z}$                  \\
    $\text{A} + \text{pH}$        & $\color{DarkGreen}0$ & $ \color{blue}0 $ & $ \color{red}\mathbb{Z}\oplus\mathbb{Z}$  
                    & $\color{DarkGreen}\mathbb{Z}\oplus\mathbb{Z} $ & $ \color{blue}0 $ & $ \color{red}0$                             
                    & $\color{DarkGreen}0 $ & $ \color{blue}0 $ & $ \color{red}\mathbb{Z}\oplus\mathbb{Z}$  \\ 
    $\text{A} + \text{SLS}$       & $\color{DarkGreen}\mathbb{Z} $ & $ \color{blue}0 $ & $ \color{red}\mathbb{Z}\oplus\mathbb{Z}$  
                    & $\color{DarkGreen}0 $ & $ \color{blue}0 $ & $ \color{red}0$                             
                    & $\color{DarkGreen}\mathbb{Z} $ & $ \color{blue}0 $ & $ \color{red}\mathbb{Z}\oplus\mathbb{Z}$  \\
    $\text{AIII}$          & $\color{DarkGreen}\mathbb{Z} $ & $ \color{blue}0 $ & $ \color{red}\mathbb{Z}\oplus\mathbb{Z}$  
                    & $\color{DarkGreen}0 $ & $ \color{blue}0 $ & $ \color{red}0$                             
                    & $\color{DarkGreen}\mathbb{Z} $ & $ \color{blue}0 $ & $ \color{red}\mathbb{Z}\oplus\mathbb{Z}$  \\
\end{tabular}
\end{ruledtabular}
\end{table*}

\begin{table*}
\caption{\label{tab:reLineIm}
\textit{K}-groups for frequency-dependent non-Hermitian Hamiltonians in $d=1,2,3$ spatial dimensions for the case of an \emph{imaginary line gap}. For each value of $d$, the three columns show the \textit{K}-groups describing the momentum dependence (green, left column), the frequency dependence (blue, middle column), and the combined momentum and frequency dependence (red, right column).}
\begin{ruledtabular}
\begin{tabular}{l@{\qquad}lll@{\qquad}lll@{\qquad}lll}
    & & $d=1$ & & & $d=2$ & & & $d=3$ & \\
    Class & $K(S^d)$ & $K(S^1)$ & $K(S^{d+1})$ & $K(S^d)$ & $K(S^1)$ & $K(S^{d+1})$ & $K(S^d)$ & $K(S^1)$ & $K(S^{d+1})$ \\ \hline     
    $\text{A}$             & $\color{DarkGreen}0 $ & $ \color{blue}0 $ & $ \color{red}\mathbb{Z}$                  
                    & $\color{DarkGreen}\mathbb{Z} $ & $ \color{blue}0 $ & $ \color{red}0$                             
                    & $\color{DarkGreen}0 $ & $ \color{blue}0 $ & $ \color{red}\mathbb{Z}$                  \\
    $\text{A} + \text{pH}$        & $\color{DarkGreen}\mathbb{Z} $ & $ \color{blue}\mathbb{Z} $ & $ \color{red}0$                             
                    & $\color{DarkGreen}0 $ & $ \color{blue}\mathbb{Z} $ & $ \color{red}\mathbb{Z}$                  
                    & $\color{DarkGreen}\mathbb{Z} $ & $ \color{blue}\mathbb{Z} $ & $ \color{red}0$                             \\
    $\text{A} + \text{SLS}$       & $\color{DarkGreen}\mathbb{Z} $ & $ \color{blue}0 $ & $ \color{red}\mathbb{Z}\oplus\mathbb{Z}$  
                    & $\color{DarkGreen}0 $ & $ \color{blue}0 $ & $ \color{red}0$                             
                    & $\color{DarkGreen}\mathbb{Z} $ & $ \color{blue}0 $ & $ \color{red}\mathbb{Z}\oplus\mathbb{Z}$  \\
    $\text{AIII}$          & $\color{DarkGreen}0 $ & $ \color{blue}0 $ & $ \color{red}\mathbb{Z}\oplus\mathbb{Z}$  
                    & $\color{DarkGreen}\mathbb{Z}\oplus\mathbb{Z} $ & $ \color{blue}0 $ & $ \color{red}0$                             
                    & $\color{DarkGreen}0 $ & $ \color{blue}0 $ & $ \color{red}\mathbb{Z}\oplus\mathbb{Z}$  \\
\end{tabular}
\end{ruledtabular}
\end{table*}

\section{Example for winding in the frequency direction}
\label{sc:example}

In this section, we consider a simple model for an open quantum system that shows nontrivial winding in the frequency direction. The model is zero dimensional in space ($d=0$). As illustrated in Fig.\ \ref{pic:model}, it consists of three orbitals for spin-$1/2$ fermions with energies $-\epsilon$, $0$, and $+\epsilon$. Only fermions in the spin eigenstate $|{\uparrow}\rangle$ satisfying $\sigma_z |{\uparrow}\rangle = |{\uparrow}\rangle$ can hop between the left and center orbitals, whereas only fermions in the spin eigenstate $|{\rightarrow}\rangle$ satisfying $\sigma_x |{\rightarrow}\rangle = |{\rightarrow}\rangle$ can hop between the right and center orbitals. $\sigma_i$ are Pauli matrices. The left and right orbitals are coupled to a bath that leads to loss, which is implemented by a constant damping rate $\Gamma>0$. The center orbital is coupled to another bath that causes energy gain, which is implemented by a constant negative damping rate $\Gamma_0<0$. We note that gain is not essential for our results in that we could remove the gain term at the expense of treating a point gap away from the point $E=0$ in the following.

\begin{figure}
\includegraphics[height=4.5cm]{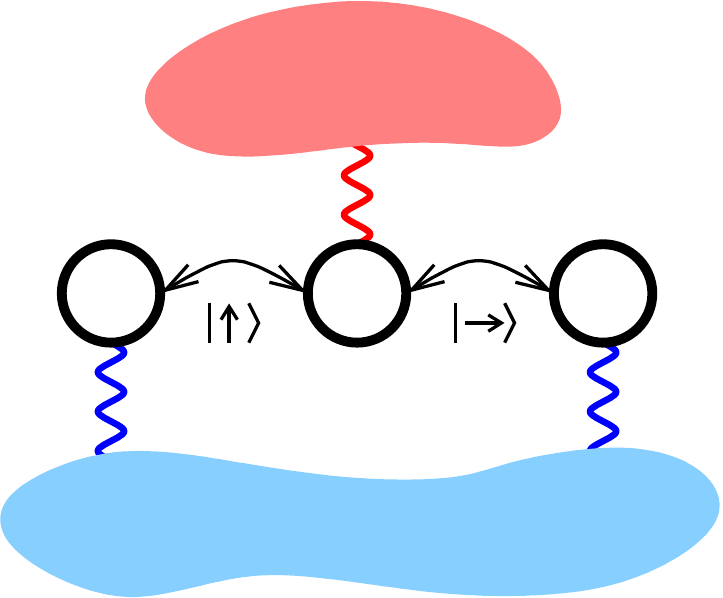}
\caption{\label{pic:model}Schematic of illustrative model consisting of three spin-full orbitals with spin-dependent hopping and coupled to baths.}
\end{figure}

We start by treating the isolated system without gain or loss. It is described by the Hamiltonian
\begin{equation}
H = H_0 + V_L + V_R ,
\end{equation}
where
\begin{equation}
H_0 = \epsilon \sum_\sigma c_{L\sigma}^\dagger c_{L\sigma}
  - \epsilon \sum_\sigma c_{R\sigma}^\dagger c_{R\sigma}
\end{equation}
describes the on-site energies and
\begin{align}
V_L &= v\, c_{L\uparrow}^\dagger d_\uparrow + \mathrm{H.c.} , \\
V_R &= v\, \frac{c_{R\uparrow}^\dagger + c_{R\downarrow}^\dagger}{\sqrt{2}}\,
  \frac{d_\uparrow + d_\downarrow}{\sqrt{2}} + \mathrm{H.c.}
\end{align}
describe the hopping. Here, $c_{l\sigma}^\dagger$ creates a fermion in orbital $l=L,R$ with spin $\sigma$ and $d_\sigma^\dagger$ creates a fermion in the center orbital with spin~$\sigma$.

Since the Hamiltonian is bilinear we can obtain the retarded Green's function $G^R_{l\sigma,l'\sigma'}(\omega)$ in closed form. Here, $l=L,0,R$, with $l=0$ corresponding to the center orbital. The equation of motion for the local Green's function for the center orbital reads as~\cite{BruusFlensberg}
\begin{align}
&(\omega + i0^+)\, G^R_{0\sigma,0\sigma'}(\omega)
  - v\, \delta_{\sigma\uparrow}\, G^R_{L\uparrow,0\sigma'}(\omega) \nonumber \\
&\quad{} - \frac{v}{2}\, \sum_{\sigma''} G^R_{R\sigma'',0\sigma'}(\omega)
  = \delta_{\sigma\sigma'} .
\label{eq:model.G0}
\end{align}
We also need the equation of motion for the nonlocal Green's functions $G^R_{l\sigma,0\sigma'}(\omega)$, which reads as
\begin{align}
&(\omega - \epsilon_n + i0^+)\, G^R_{n\sigma,0\sigma'}(\omega)
  - v\, \delta_{nL} \delta_{\sigma\uparrow}\, G^R_{0\uparrow,0\sigma'}(\omega) \nonumber \\
&\quad{} - \frac{v}{2}\, \delta_{nR} \sum_{\sigma''} G^R_{0\sigma'',0\sigma'}(\omega) = 0 ,
\label{eq:model.GLR}
\end{align}
where $\epsilon_L = \epsilon$ and $\epsilon_R = -\epsilon$. Solving Eq.\ (\ref{eq:model.GLR}) for the nonlocal Green's functions and inserting the result into Eq.\ (\ref{eq:model.G0}), we obtain
\begin{align}
&(\omega + i0^+)\, G^R_{0\sigma,0\sigma'}(\omega)
  - \delta_{\sigma\uparrow}\, \frac{v^2}{\omega - \epsilon + i0^+}\,
    G^R_{0\uparrow,0\sigma'}(\omega) \nonumber \\
&\quad{} - \frac{v^2/4}{\omega + \epsilon + i0^+} \sum_{\sigma''} G^R_{0\sigma'',0\sigma'}(\omega)
  = \delta_{\sigma\sigma'} .
\end{align}
Writing the Green's function as a $2\times 2$ matrix acting on spin space, this equation becomes
\begin{align}
&(\omega + i0^+)\, G^R_{00}(\omega)
  - |{\uparrow}\rangle\langle{\uparrow}|\, \frac{v^2}{\omega - \epsilon + i0^+}\,
    G^R_{00}(\omega) \nonumber \\
&\quad{} - |{\rightarrow}\rangle\langle{\rightarrow}|\, \frac{v^2}{\omega + \epsilon + i0^+}\,
    G^R_{00}(\omega) = I ,
\end{align}
where $I$ is the identity matrix. The solution, with suppressed identity matrices, reads as
\begin{equation}
G^R_{00}(\omega) = \left[ \omega - H^R_\mathrm{eff}(\omega) \right]^{-1} ,
\end{equation}
with the effective Hamiltonian
\begin{equation}
H^R_\mathrm{eff}(\omega)
  = |{\uparrow}\rangle\langle{\uparrow}|\, \frac{v^2}{\omega - \epsilon + i0^+}
  + |{\rightarrow}\rangle\langle{\rightarrow}|\, \frac{v^2}{\omega + \epsilon + i0^+}
  - i0^+ .
\end{equation}
Next, we introduce gain and loss by replacing the infinitesimal imaginary terms by finite ones:
\begin{equation}
H^R_\mathrm{eff}(\omega)
  = |{\uparrow}\rangle\langle{\uparrow}|\, \frac{v^2}{\omega - \epsilon + i\Gamma}
  + |{\rightarrow}\rangle\langle{\rightarrow}|\, \frac{v^2}{\omega + \epsilon + i\Gamma}
  - i\Gamma_0 ,
\label{eq:model.Heff}
\end{equation}
where $\Gamma>0$ (loss) and $\Gamma_0<0$ (gain).

One easily checks that the effective Hamiltonian in Eq.\ (\ref{eq:model.Heff}) satisfies the symmetries $\mathrm{TRS}^\dagger$ with the matrix $U_{\mathrm{TRS}^\dagger} = I$, $\mathrm{PHS}^\dagger$ with the matrix~\footnote{This matrix describes a rotation by $\pi$ about the $[101]$ direction in spin space, which interchanges the $x$ and $z$ components.}
\begin{equation}
U_{\mathrm{PHS}^\dagger} = \exp\left( - i\pi\, \frac{\sigma_1 + \sigma_3}{2\sqrt{2}} \right) ,
\end{equation}
and $\mathrm{CS}$ with the same matrix $U_\mathrm{CS} = U_{\mathrm{PHS}^\dagger}$. The anti-unitary symmetry transformations for $\mathrm{TRS}^\dagger$ and $\mathrm{PHS}^\dagger$ both square to $+I$. Table \ref{tab:symClasses} shows that the model thus belongs to class $\mathrm{BDI}^\dagger$.

\begin{figure}
\includegraphics[width=\columnwidth]{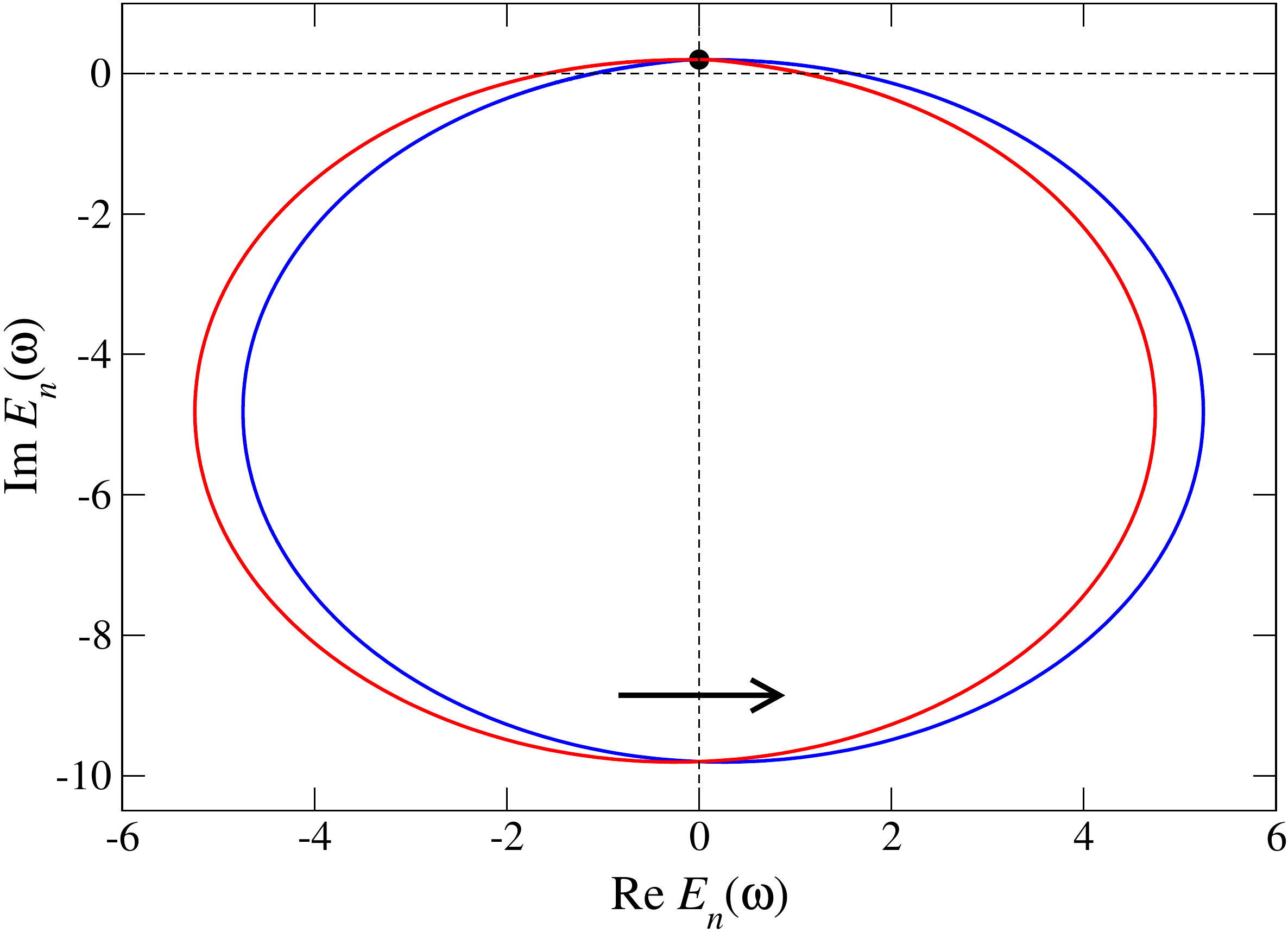}
\caption{\label{pic:model.eigenvalues}Parametric plot of the eigenvalues $E_1(\omega)$ (blue curve) and $E_2(\omega)$ (red curve) of the effective Hamiltonian $H^R_\mathrm{eff}(\omega)$ in the complex plane. The frequency $\omega$ runs from $-\infty$, where $E_1 = E_2 = -i\Gamma_0$ (black dot), to $+\infty$, with the same values of $E_1$ and $E_2$. Both curves are traversed counterclockwise, as indicated by the arrow. The parameters are $\epsilon = 1$, $v = 1$, $\Gamma_0 = -0.2$, and $\Gamma = 0.1$.}
\end{figure}

The complex eigenvalues of $H^R_\mathrm{eff}(\omega)$ are plotted in Fig.\ \ref{pic:model.eigenvalues}. Both eigenvalues orbit the zero point in the positive direction once, as $\omega$ increases from $-\infty$ to $+\infty$. Hence, the model has a point gap. According to Table \ref{tab:rePoint}, the \textit{K}-group describing the frequency winding for class $\mathrm{BDI}^\dagger$ is $\mathbb{Z}$. The invariant is easily identified as a winding number
\begin{equation}
w = \frac{1}{2\pi i} \int d\omega\, \frac{\partial}{\partial\omega}
  \ln \det H^R_\mathrm{eff}(\omega)
\end{equation}
The integrand can be obtained analytically. Evaluating the integral numerically, we obtain $w=2$. In fact, the winding number can equivalently be written as~\cite{PhysRevX.8.031079}
\begin{equation}
w = \frac{1}{2\pi} \sum_n \int d\omega\, \frac{\partial}{\partial\omega} \arg E_n(\omega) .
\end{equation}
With this, one finds $w=2$ directly by inspecting Fig.~\ref{pic:model.eigenvalues}.

The crucial ingredient for the nontrivial topology of this model is the incompatibility of the hopping operators $V_L$ and $V_R$. Together with the different on-site energies of the left and right orbitals, $V_L$ and $V_R$ generate two incompatible self-energy corrections that peak at different frequencies $\omega$, see Eq.\ (\ref{eq:model.Heff}). This causes the eigenvalues to describe closed loops in the complex energy plane.

\section{Summary and conclusions}
\label{sc:concl}

Non-Hermitian Hamiltonians that depend on frequency as well as on momentum arise, for example, as effective Hamiltonians that are in one-to-one correspondence to single-particle Green's functions. They are thus well defined also for interacting systems. In this paper, we have focused on the additional topological features that result from the frequency dependence. Our results rely on the assumption that the frequency axis can be compactified to a circle $S^1$, which requires the effective Hamiltonian and thus the self-energy to have the same limit for $\omega\to\pm\infty$. This holds for example if the effect of interactions becomes weak at high energies.

We have first presented a rigorous definition and calculation of the group structure formed by stable equivalence classes of matrix-valued functions form $S^{d+1}$ with additional symmetry and gap constraints. Important theorems have been reformulated in the transparent language of Hamiltonians instead of vector bundles. An essential point is that adding flat bands allows decomposing nontrivial parts of a Hamiltonians into irreducible parts forming a group. The overall \textit{K}-group consists of three parts, one from the pure momentum dependence, one from the frequency dependence, and one from the combination of the two. The frequency dependence thus adds two contribution: a contribution that describes possible winding of the effective Hamiltonian in the frequency direction and a contribution that describes winding in the $(d+1)$-dimensional momentum-frequency space. All groups are found to be isomorphic to direct sums of potentially multiple copies of the groups $\mathbb{Z}$ and $\mathbb{Z}_2$.
An example with nontrivial winding in the frequency direction is given by Ramos \textit{et al.}\ \cite{Ramos_2021} in the context of an effective \emph{Hermitian} Hamiltonian resulting from the quantum Langevin equation.

Frequency and non-Hermiticity turn out to lead to $54$ symmetry classes, which are listed in Table \ref{tab:symClasses}. This 54-fold way has to be contrasted to the tenfold way \cite{Zirnbauer, PhysRevB.55.1142} for Hermitian Hamiltonians with only momentum dependence and the 38-fold way \cite{PhysRevX.9.041015} for non-Hermitian Hamiltonians with momentum dependence. The lack on Hermiticity also means that the spectrum is complex, which requires a reconsideration of the meaning of a spectral gap. We have considered the cases of a point gap, a real line gap, and an imaginary line gap \cite{PhysRevX.9.041015}. We have applied a flattening procedure, which maps the problem to one concerning Hermitian Hamiltonians for each of the three gap types, and dimensional reduction to obtain the group structure for each symmetry class, gap type, and spatial dimensions $d=1,2,3$. The dimension-reduction procedure differs from the one used earlier \cite{Teo_2010, PhysRevB.90.165114} to allow us to deal with symmetry transformations that invert frequency. The results are presented in Tables \ref{tab:rePoint}--\ref{tab:reLineIm}, which contain the tenfold-way classification \cite{Schnyder_2008, Schnyder_2009, Kitaev_2009, Ryu_2010, Ludwig_2015} as well as the classification in terms of 38 symmetry classes of only momentum-dependent non-Hermitian Hamiltonians~\cite{PhysRevX.9.041015}.

This work opens several directions for future research. Our results establish the possible existence and nature of topological invariants for many cases but do not say how these can be calculated. Expressions for these invariants are called for. Clearly, one would also like to find models and experimental realizations belonging to entries in the tables. Cases where the frequency dependence and the non-Hermiticity are essential are particularly interesting. Moreover, what are the observable consequences of the nontrivial \textit{K}-groups? From our experience with topological insulators and fully gapped superconductors, one would of course primarily look for surface states. Such states at surfaces in real space have a clear meaning but our Hamiltonians also have a frequency argument, which suggests to consider end states in time. We speculate that sudden switching (quenches) between different symmetry classes could lead to nontrivial transient effects. At a more general level, there is a wide field of study of possible topological phase transitions.

Furthermore, it is clear that the effective Hamiltonian obtained from the single-particle Green's function does not contain everything that there is to know about an interacting system. One can think of applying analogous ideas to classify higher-order correlation and response functions.

\begin{acknowledgments}
The authors thank J. Budich and S. Kobayashi for useful discussions.
Financial support by the Deutsche Forschungsgemeinschaft through Collaborative Research Center SFB 1143, project A04, project id 247310070, and the W\"urzburg-Dresden Cluster of Excellence ct.qmat, EXC 2147, project id 390858490, is gratefully acknowledged.
\end{acknowledgments}

\appendix

\section{Symmetry transformations of the frequency}
\label{app:SymFreq}

In this appendix, we consider the possible sign change of the frequency argument $\omega$ of the Green's function $G^R(\mathbf{k},\omega)$. We consider noninteracting systems and assume, as always, translational invariance on a lattice. The Hamiltonian on Fock space is then bilinear and can be written as
\begin{equation}
\mathcal{H}_0 = \sum_\mathbf{k} \sum_{nn'} a^\dagger_{\mathbf{k}n} H_{nn'}(\mathbf{k}) a_{\mathbf{k}n'} ,
\end{equation}
where $a_{\mathbf{k}n}$ and $a^\dagger_{\mathbf{k}n}$ are annihilation and creation operators for particles with fixed momentum $\mathbf{k}$ and further quantum numbers collected in $n$. The retarded Green's function is a matrix with components
\begin{equation}
G^R_{nn'}(\mathbf{k},\omega) = -i \int_0^\infty dt\, e^{i\omega t}
  \langle\mathrm{GS}| \big\{ a_{\mathbf{k}n}(t),
  a^\dagger_{\mathbf{k}n'}(0) \big\} |\mathrm{GS}\rangle ,
\end{equation}
where $|\mathrm{GS}\rangle$ is the ground state and we employ the Heisenberg picture. This is the general form. For the noninteracting case, the Green's function becomes
\begin{equation}
G^R(\mathbf{k},\omega) = \big[ \omega - H(\mathbf{k}) + i0^+ \big]^{-1} .
\end{equation}
Identity matrices are suppressed.

\subsection{Unitary symmetries}

Let $\mathcal{U}$ be a unitary operator on Fock space that describes an ordinary unitary point-group symmetry. This means that $\mathcal{U}$ acts on annihilation operators as
\begin{equation}
\mathcal{U} a_{\mathbf{k}n} \mathcal{U}^{-1} = \sum_{n'} U^\dagger_{nn'} a_{R\mathbf{k},n'}
  = \sum_{n'} U^*_{n'n} a_{R\mathbf{k},n'} ,
\end{equation}
where $U$ is a unitary matrix and $R \in \mathrm{O}(3)$ describes the action of the symmetry transformation on the momentum vector $\mathbf{k}$. It is useful to also write down the transformation rules for creation operators and for the inverse transformations:
\begin{align}
\mathcal{U} a^\dagger_{\mathbf{k}n} \mathcal{U}^{-1}
  &= \sum_{n'} U^T_{nn'} a^\dagger_{R\mathbf{k},n'}
  = \sum_{n'} U_{n'n} a_{R\mathbf{k},n'} , \\
\mathcal{U}^{-1} a_{\mathbf{k}n} \mathcal{U} &= \sum_{n'} U_{nn'} a_{R^{-1}\mathbf{k},n'} , \\
\mathcal{U}^{-1} a^\dagger_{\mathbf{k}n} \mathcal{U}
  &= \sum_{n'} U^*_{nn'} a^\dagger_{R^{-1}\mathbf{k},n'} .
\end{align}
To be a symmetry, $\mathcal{U}$ must satisfy $\mathcal{U} \mathcal{H} \mathcal{U}^{-1} = \mathcal{H}$ so that
\begin{equation}
U H(R^{-1}\mathbf{k}) U^\dagger = H(\mathbf{k}) .
\end{equation}
The Green's function then satisfies
\begin{equation}
U G^R(R^{-1}\mathbf{k},\omega) U^\dagger = G^R(\mathbf{k},\omega) .
\end{equation}
For the noninteracting case, this follows trivially but it holds in general. In the general proof, we have to assume that the ground state $|\mathrm{GS}\rangle$ is invariant under $\mathcal{U}$ up to a phase factor. This means that the result breaks down if the symmetry $\mathcal{U}$ is spontaneously broken.

\subsection{particle-hole symmetry}

On Fock space, the particle-hole transformation or charge conjugation $\mathcal{C}$ is \emph{unitary} \cite{Schnyder_2008, Ryu_2010}. However, it differs from an ordinary symmetry by its action on annihilation and creation operators:
\begin{align}
\mathcal{C} a_{\mathbf{k}n} \mathcal{C}^{-1} &= \sum_{n'} U^T_{C,nn'} a^\dagger_{-\mathbf{k},n'}
  = \sum_{n'} U_{C,n'n} a^\dagger_{-\mathbf{k},n'} ,
\label{C.a1} \\
\mathcal{C} a^\dagger_{\mathbf{k}n} \mathcal{C}^{-1}
  &= \sum_{n'} U^\dagger_{C,nn'} a_{-\mathbf{k},n'}
  = \sum_{n'} U^*_{C,n'n} a_{-\mathbf{k},n'} ,
\label{C.a2} \\
\mathcal{C}^{-1} a_{\mathbf{k}n} \mathcal{C} &= \sum_{n'} U_{C,nn'} a^\dagger_{-\mathbf{k},n'} , \\
\mathcal{C}^{-1} a^\dagger_{\mathbf{k}n} \mathcal{C}
  &= \sum_{n'} U^*_{C,nn'} a_{-\mathbf{k},n'} .
\end{align}
This is only possible for fermions. The reason is that for fermions, but not for bosons, creation and annihilation operators have the same algebraic properties and therefore the particle-hole transformation is possible.

To be a symmetry, $\mathcal{C}$ must satisfy $\mathcal{C} \mathcal{H} \mathcal{C}^{-1} = \mathcal{H}$. Note that there is no minus sign; this is a standard unitary symmetry on Fock space. Then for the noninteracting case we find that two conditions must be satisfied~\cite{Schnyder_2008, Ryu_2010, Schnyder_2012}:
\begin{align}
\sum_\mathbf{k} \Tr H(\mathbf{k}) &= 0 , \\
U_C H(-\mathbf{k})^* U_C^\dagger &= - H(\mathbf{k}) .
\end{align}
For Hermitian Hamiltoniuans, we can replace $H(-\mathbf{k})^*$ by $H(-\mathbf{k})^T$ as convenient. The Green's function satisfies
\begin{equation}
U_C G^R(-\mathbf{k},-\omega)^* U_C^\dagger = - G^R(\mathbf{k},\omega) .
\end{equation}

\subsection{Time-reversal symmetry}

Time reversal $\mathcal{T}$ is an \emph{anti-unitary} transformation on Fock space, i.e., $\mathcal{T} i \mathcal{T}^{-1} = -i$, which acts on annihilation and creation operators as
\begin{align}
\mathcal{T} a_{\mathbf{k}n} \mathcal{T}^{-1} &= \sum_{n'} U^\dagger_{T,nn'} a_{-\mathbf{k},n'}
  = \sum_{n'} U^*_{T,n'n} a_{-\mathbf{k},n'} , \\
\mathcal{T} a^\dagger_{\mathbf{k}n} \mathcal{T}^{-1}
  &= \sum_{n'} U^T_{T,nn'} a^\dagger_{-\mathbf{k},n'}
  = \sum_{n'} U_{T,n'n} a^\dagger_{-\mathbf{k},n'} , \\
\mathcal{T}^{-1} a_{\mathbf{k}n} \mathcal{T} &= \sum_{n'} U^*_{T,nn'} a_{-\mathbf{k},n'} , \\
\mathcal{T}^{-1} a^\dagger_{\mathbf{k}n} \mathcal{T}
  &= \sum_{n'} U_{T,nn'} a^\dagger_{-\mathbf{k},n'} .
\end{align}
To be a symmetry, $\mathcal{T}$ must satisfy $\mathcal{T} \mathcal{H} \mathcal{T}^{-1} = \mathcal{H}$. This is a nonstandard symmetry since $\mathcal{T}$ is anti-unitary. Then for the noninteracting case we find~\cite{Schnyder_2008, Ryu_2010, Schnyder_2012}
\begin{equation}
U_T H(-\mathbf{k})^T U_T^\dagger = H(\mathbf{k}) .
\end{equation}
For Hermitian Hamiltonians, we can replace $H(-\mathbf{k})^T$ by $H(-\mathbf{k})^*$. The Green's function satisfies~\cite{PhysRevLett.105.256803}
\begin{equation}
U_T G^R(-\mathbf{k},\omega)^T U_T^\dagger = G^R(\mathbf{k},\omega) .
\end{equation}

\subsection{Chiral symmetry}

Chiral or sublattice symmetry is implemented by $\mathcal{S} = \mathcal{CT}$, which is anti-unitary, $\mathcal{S} i \mathcal{S}^{-1} = -i$. All results follow from the previous two sections. Chiral symmetry acts on annihilation and creation operators as
\begin{align}
\mathcal{S} a_{\mathbf{k}n} \mathcal{S}^{-1} &= \sum_{n'} U^T_{S,nn'} a^\dagger_{\mathbf{k}n'}
  = \sum_{n'} U_{S,n'n} a^\dagger_{\mathbf{k}n'} , \\
\mathcal{S} a^\dagger_{\mathbf{k}n} \mathcal{S}^{-1}
  &= \sum_{n'} U^\dagger_{S,nn'} a_{\mathbf{k}n'}
  = \sum_{n'} U^*_{S,n'n} a_{\mathbf{k}n'} , \\
\mathcal{S}^{-1} a_{\mathbf{k}n} \mathcal{S} &= \sum_{n'} U^*_{S,nn'} a^\dagger_{\mathbf{k}n'} , \\
\mathcal{S}^{-1} a^\dagger_{\mathbf{k}n} \mathcal{S}
  &= \sum_{n'} U_{S,nn'} a_{\mathbf{k}n'} .
\end{align}
Here, $U_S = U_C U_T^*$. To be a symmetry, $\mathcal{S}$ must satisfy $\mathcal{S} \mathcal{H} \mathcal{S}^{-1} = H$. Then for the noninteracting case we find~\cite{Schnyder_2008, Ryu_2010, Schnyder_2012}
\begin{align}
\sum_\mathbf{k} \Tr H(\mathbf{k}) &= 0 , \\
U_S H(\mathbf{k}) U_S^\dagger &= - H(\mathbf{k}) .
\end{align}
The Green's function satisfies
\begin{equation}
U_S G^R(\mathbf{k},-\omega)^\dagger U_S^\dagger = - G^R(\mathbf{k},\omega) .
\end{equation}
This can also be written as
\begin{equation}
U_S G^A(\mathbf{k},-\omega) U_S^\dagger = - G^R(\mathbf{k},\omega) .
\end{equation}

\section{Zero-dimensional Altland-Zirnbauer classes}
\label{sc:zeroAZ}

In this appendix, we recall the construction of classifying spaces for frequency-independent Hermitian Hamiltonians within one of the Altland-Zirnbauer (AZ) classes \cite{Kitaev_2009, PhysRevB.85.085103, PhysRevB.88.125129}. Such Hamiltonians satisfy a subset of the symmetries from Table \ref{tab:TabSymHer}, also compare Table \ref{tab:SymW}. For Hermitian Hamiltonians, we of course have $H=H^\dagger$ and $H^T=H^*$. The definition of the classes and the required symmetry operations can be found for example in Refs.~\cite{Ryu_2010, Ludwig_2015}.

\begin{table}[hbt]
\caption{\label{tab:TabSymHer}Symmetry operations for frequency-independent Hermitian Hamiltonians. $K$ denotes complex conjugation.}
\begin{ruledtabular}
\begin{tabular}{lll}
TRS & $T\, H(\mathbf{k})\, T^{-1} = H(-\mathbf{k})$ & $T = U_\mathrm{TRS}K$ \\
PHS & $C\, H(\mathbf{k})\, C^{-1} = -H(-\mathbf{k})$ & $C = U_\mathrm{PHS}K$ \\
SLS & $S\, H(\mathbf{k})\, S^{-1} = -H(\mathbf{k})$ & $S = TC = U_\mathrm{TRS} U_\mathrm{PHS}^*$
%
%
\end{tabular}
\end{ruledtabular}
\end{table}

The goal is to construct the classifying space for Hamiltonians $H$ with certain of these symmetries as described by one of the AZ classes, in $D=0$ dimensions. The generators of these symmetries realize a complex or real Clifford algebra $\{e_1,\ldots\}$. The classifying space is obtained from the extension of this Clifford algebra by another generator $e_0$, which is judiciously chosen as the Hamiltonian $H$ or as $H$ multiplied by the imaginary unit.
The complex case is straightforward and the results are shown in Table~\ref{tab:0DimComplAZ}.

\begin{table}[hbt]
\caption{\label{tab:0DimComplAZ}
Classifying spaces for the complex AZ classes in $D=0$ dimensions~\cite{PhysRevB.88.125129}.}
\begin{ruledtabular}
\begin{tabular}{llll}
    Class   &  Generators           & Extension                 & Classifying space  \\ \hline
    A       &  $e_0=H$              & $\Cl_0 \rightarrow \Cl_1$   & $\mathcal{C}_0$ \\
    AIII    &  $e_0=H$, $e_1=S$     & $\Cl_1 \rightarrow \Cl_2$   & $\mathcal{C}_1$
  \end{tabular}
\end{ruledtabular}
\end{table}

In the real case, there are three possibilities:
If only TRS is present we have $e_1=T$ and $e_2=T J$ satisfying $e_1^2 = e_2^2 = \epsilon_T = \pm 1$ and extend the algebra by $e_0 = J H$ with $e_0^2 = -1$. Here, $J$ is a real matrix representation of the imaginary unit.
If only PHS is present we have $e_1 = C$ and $e_2=C J$ satisfying $e_1^2 = e_2^2 = \epsilon_C = \pm 1$ and extend the algebra by $e_0 = H$ with $e_0^2 = +1$.
If both symmetries are present we have $e_1 = C$, $e_2 = C J$, and $e_3 = TC J$ satisfying $e_1^2 = e_2^2 = \epsilon_C = \pm 1$ and $e_3^2 = -\epsilon_T \epsilon_C$ and extend the algebra by $e_0 = H$ with $e_0^2 = +1$.
This choice of generators leads to the results presented in Table \ref{tab:0DimRealAZ}.
In the last three rows, Bott periodicity $\mathcal{R}_p \approx \mathcal{R}_{p+8}$ is used.

\begin{table}[hbt]
\caption{\label{tab:0DimRealAZ}
Classifying spaces for the real AZ classes in $D=0$ dimensions~\cite{PhysRevB.88.125129}.}
\begin{ruledtabular}
  \begin{tabular}{l l l l}
    Class   &  $(\epsilon_{T},\epsilon_{C})$    & Extension                         & Classifying space  \\ \hline
    AI      &  $(1,0)$                          & $\Cl_{0,2} \rightarrow \Cl_{1,2}$   & $\mathcal R_0$ \\
    BDI     &  $(1,1)$                          & $\Cl_{1,2} \rightarrow \Cl_{1,3}$   & $\mathcal R_1$ \\
    D       &  $(0,1)$                          & $\Cl_{0,2} \rightarrow \Cl_{0,3}$   & $\mathcal R_2$ \\
    DIII    &  $(-1,1)$                         & $\Cl_{0,3} \rightarrow \Cl_{0,4}$   & $\mathcal R_3$ \\
    AII     &  $(-1,0)$                         & $\Cl_{2,0} \rightarrow \Cl_{3,0}$   & $\mathcal R_4$ \\
    CII     &  $(-1,-1)$                        & $\Cl_{3,0} \rightarrow \Cl_{3,1}$   & $\mathcal R_{-3} \approx \mathcal R_5$ \\
    C       &  $(0,-1)$                         & $\Cl_{2,0} \rightarrow \Cl_{2,1}$   & $\mathcal R_{-2} \approx \mathcal R_6$ \\
    CI      &  $(1,-1)$                         & $\Cl_{2,1} \rightarrow \Cl_{2,2}$   & $\mathcal R_{-1} \approx \mathcal R_7$ 
  \end{tabular}
\end{ruledtabular}
\end{table}

\bibliography{Kotz}

\end{document}